\RequirePackage{fix-cm}
\documentclass[smallextended]{svjour3}       
\smartqed  
\usepackage{color,graphicx,verbatim,amsmath,amsfonts, url,cases,tikz,epsfig}
\graphicspath{{ArXiv/Plots/}}
\usepackage[numbers,sort&compress]{natbib}

\newcommand{\PP}{\mathbb{P}}
\newcommand{\EE}{\mathbb{E}}
\newcommand{\RR}{\mathbb{R}}
\newcommand{\R}{\mathbb{R}}
\newcommand{\dd}{\mathrm{d}}
\newcommand{\ind}{\boldsymbol{1}}
\newcommand{\bZ}{\boldsymbol{Z}}

\newcommand{\bD}{\boldsymbol{D}}

\newcommand{\br}{\boldsymbol{r}}
\newcommand{\bv}{\boldsymbol{v}}
\newcommand{\bx}{\boldsymbol{x}}
\newcommand{\by}{\boldsymbol{y}}
\newcommand{\bu}{\boldsymbol{u}}
\newcommand{\boldf}{\boldsymbol{f}}
\newcommand{\bs}{\boldsymbol{s}}
\newcommand{\bzeta}{\boldsymbol{\zeta}}
\newcommand{\btheta}{\boldsymbol{\theta}}
\newcommand{\bPsi}{\boldsymbol{\Psi}}
\newcommand{\bolde}{\boldsymbol{e}}
\newcommand{\bxi}{\boldsymbol{\xi}}
\newcommand{\origin}{\boldsymbol{0}}
\numberwithin{equation}{section}
\numberwithin{figure}{section}

\newcommand{\EL}{\mathcal{L}}
\newcommand{\deltain}{\delta_\text{in}}
\newcommand{\bzero}{\boldsymbol{0}}

\newcommand{\Din}{D^\text{in}}
\newcommand{\Dout}{D^\text{out}}

\newcommand{\convas}{\stackrel{\text{a.s.}}{\longrightarrow}}
\newcommand{\convp}{\stackrel{p}{\longrightarrow}}

\newtheorem{Theorem}{Theorem}
\newtheorem{Corollary}{Corollary}
\newtheorem{Proposition}{Proposition}
\newtheorem{Definition}{Definition}

\allowdisplaybreaks
\begin{document}

\title{Asymptotic Dependence of In- and Out-Degrees in a Preferential Attachment Model with Reciprocity
}

\titlerunning{Preferential Attachment with Reciprocity}        

\author{Tiandong Wang        \and
        Sidney I. Resnick 
}


\institute{Tiandong Wang \at
              Department of 
	Statistics, Texas 
	A\&M University, College 
		Station, TX 77843, U.S. \\
              \email{twang@stat.tamu.edu}           
           \and
           Sidney I. Resnick \at
	School of Operations Research and Information Engineering,
	 Cornell University, Ithaca, NY 14853, U.S. 
		\email{sir1@cornell.edu}
}

\date{Received: date / Accepted: date}

\maketitle

\begin{abstract}
Reciprocity characterizes the information exchange between users in a network, and some empirical studies have revealed that social networks have a high proportion of reciprocal edges. Classical directed preferential attachment (PA) models, though generating scale-free networks, 
may give networks with low reciprocity. This points out one potential problem of fitting a classical PA model to a given network dataset with high reciprocity, and indicates alternative models need to be considered. We give one possible modification of the classical PA model by including another parameter which controls the probability of adding a reciprocated edge at each step. 
Asymptotic analyses suggest that large in- and out-degrees become fully dependent in this modified model, as a result of the additional reciprocated edges.
\keywords{
Reciprocity \and multivariate regular variation \and asymptotic dependence \and preferential attachment
}
\subclass{MSC 
05C82 \and MSC 60F15 \and MSC 60G70\and MSC 62G32
}
\end{abstract}

\section{Introduction}
The \emph{reciprocity} coefficient, which is classically defined as the proportion of reciprocated edges (cf. \cite{newman:etal:2001, wasserman:faust:1994}), characterizes the mutual exchange of information between users in a social network.
For instance, on Twitter, users share links or pass on information to 
a target user through sending direct @-messages. When interacting with peers on Twitter, the direct @-messages will be exchanged between users in both directions (cf. \cite{chengetal:2011}). 
In addition, a study on eight different types of networks in \cite{jiangetal:2015} shows online social networks (e.g. \cite{facebook:2009, flickr:2009, java:etal:2007, googleplus:2012,social:2007}) tend to have a higher proportion of reciprocal edges, compared to other types of networks such as biological networks, communication networks, software call graphs and P2P networks. 

To model the evolution of social networks, one appealing approach is the directed preferential attachment (PA) model (cf. \cite{bollobas:borgs:chayes:riordan:2003, krapivsky:redner:2001}), which captures the scale-free property, i.e. both in- and out-degree distributions have Pareto-like tails
(cf. \cite{resnick:samorodnitsky:towsley:davis:willis:wan:2016, resnick:samorodnitsky:2015, wan:wang:davis:resnick:2017,wang:resnick:2019}). 
Asymptotic behaviors for the reciprocity coefficient in a directed PA model have recently been developed in \cite{wang:resnick:2021}.
We see from the theoretical results in \cite{wang:resnick:2021} that for certain choice of parameters, the classical PA model will generate networks with proportion of reciprocal edges close to 0, which deviates from the empirical findings for social networks described in \cite{jiangetal:2015}. 
This motivates us to think of variants of the classical PA model that are able to generate a high proportion of reciprocated edges. 

In \cite{das:resnick:2017}, the authors analyze the Facebook wall post dataset available at 
\url{http://konect.cc/networks/facebook-wosn-wall/}, and find that for the subset of nodes with large in- and out-degrees, the in- and out-degree pairs are modeled as
highly dependent.
Taking a closer look at the data, we see more than 60\% of the edges in the Facebook wall post network are reciprocal. These facts lead to the conjecture that 
nodes with large in- and out-degrees are likely to have the in- and out-degrees
highly dependent in a network with high reciprocity.
With this conjecture in mind, we extend the directed PA model by introducing another parameter, $\rho\in (0,1)$, which
controls the probability of adding a reciprocated edge to a newly created edge following the PA rule. 
We then study the theoretical properties of this modified PA model
with emphasis on the asymptotic dependence structure of large in- and out-degrees.
Theoretical results are obtained through extending the embedding techniques in \cite{wang:resnick:2018,wang:resnick:2019} to continuous-time multitype branching processes with immigration. 
{Empirical frequencies of nodes with in-degree $i$ and out-degree $j$
  converge to a limiting regularly varying distribution $p_{ij}$, and unlike the
  traditional PA model where the mass of the limit distribution is
  spread out in the first quadrant, for the model with reciprocity}
we find that large in- and out-degrees become asymptotically fully dependent.

In Section~\ref{sec:model}, we give a detailed description of our modified model with reciprocity probability $\rho\in (0,1)$. We then present important preliminaries in Section~\ref{sec:prelim}, including background knowledge on continuous-time multitype branching processes and multivariate regular variation.
Section~\ref{sec:res} contains three theoretical properties for the modified PA model: (1) the growth of in- and out-degrees for a fixed node; (2) the asymptotic limit for the joint in- and out-degree distribution; and (3) the asymptotic dependence structure for large in- and out-degrees.
We highlight some future research directions in Section~\ref{sec:discuss},
and technical proofs are given in Section~\ref{sec:proof}.

\subsection{Model Setup}\label{sec:model}
Initialize the model with graph $G(0)$, which consists of one node (labeled as Node 1) and a self-loop.
Let $G(n)$ denote the graph  after $n$ steps and
 $V(n)$ be the set of nodes in $G(n)$ with $V(0) = \{1\}$ and $|V(0)| = {1}$.
Denote the set of directed edges in $G(n)$ by $E(n)$ such that
an ordered pair $(w_1,w_2)\in E(n)$, $w_1,w_2\in V(n)$, 
represents a directed edge $w_1\mapsto w_2$. 
When $n=0$, we have $E(0) = \{(1,1)\}$.

Set $\bigl(\Din_w(n), \Dout_w(n)\bigr)$ to
be the in- and out-degrees of node $w\in V(n)$. 
We use the convention that $\Din_w(n) = \Dout_w(n) = 0$ if $w\notin V(n)$.
From $G(n)$ to $G(n+1)$, one of the following scenarios happens:
\begin{enumerate}
\item[(i)] With probability $\alpha$, 
we add a new node $|V(n)|+1$ with a directed edge $(|V(n)|+1,w)$, where $w\in V(n)$ is chosen with probability
\begin{equation}\label{eq:Din}
\frac{\Din_w(n)+\delta}{\sum_{w\in V(n)} (\Din_w(n)+\delta)}
= \frac{\Din_w(n)+\delta}{|E(n)|+\delta |V(n)|},
\end{equation}
and update the node set $V(n+1)=V(n)\cup\{|V(n)|+1\}$.
If, with probability $\rho\in (0,1)$, a reciprocal edge $(w, |V(n)|+1)$ is added, we
 update the edge set as
$E(n+1) = E(n)\cup \{(|V(n)|+1,w), (w,|V(n)|+1)\}$. 
If the reciprocal edge is not created, set $E(n+1) = E(n)\cup \{(|V(n)|+1,w)\}$. 
\item[(ii)] With probability $\gamma$, 
we add a new node $|V(n)|+1$ with a directed edge $(w,|V(n)|+1)$, where $w\in V(n)$ is chosen with probability
\begin{equation}\label{eq:Dout}
\frac{\Dout_w(n)+\delta}{\sum_{w\in V(n)} (\Dout_w(n)+\delta)}
= \frac{\Dout_w(n)+\delta}{|E(n)|+\delta |V(n)|},
\end{equation}
and update the node set $V(n+1)=V(n)\cup\{|V(n)|+1\}$.
If, with probability $\rho\in (0,1)$, a reciprocal edge $(|V(n)|+1,w)$ is added, we
update the edge set as 
$E(n+1) = E(n)\cup \{(|V(n)|+1,w), (w,|V(n)|+1)\}$. 
If the reciprocal edge is not created, set $E(n+1) = E(n)\cup \{(w,|V(n)|+1)\}$. 
\end{enumerate}
Here we assume the offset parameter $\delta\ge 0$, and takes the same value for both in- and out-degrees.
From the description above, we see that $|V(n)|=n+1$, and
as $n\to\infty$,
\[
\frac{|E(n)|}{n}\convas 1+\rho.
\]


\section{Preliminaries}\label{sec:prelim}
We start by summarizing important embedding techniques that lay the foundation for theoretical results in Section~\ref{sec:res}. Also, to prepare for analyses on the asymptotic structure, we provide useful definitions related to multivariate regular variation. 

\subsection{Embedding techniques}\label{subsec:embed}
In this section, we first provide important background knowledge on continuous-time multitype
branching processes with immigration. Then we give the embedding framework for the
in- and out-degree sequences in our modified PA model with reciprocity. These are the key ingredients to prove the asymptotic results in Section~\ref{sec:res}.


\subsubsection{Continuous-Time Multitype Branching Process}\label{subsec:mBI}
We start this section with the general setup of continuous-time multitype branching processes as introduced in \cite[Chapter V]{athreya:ney:1972}.
Consider a set of particles of $K$ types, and a type $i$ particle has an exponential lifetime with rate $a_i>0$, $i=1,\ldots,K$. Upon its death, this type $i$ particle produces $Y^{(i)}_j$ copies of type $j$ particles for $j=1,\ldots,K$. Here $(Y^{(i)}_1,\ldots, Y^{(i)}_K)\in \mathbb{N}^K$ is a random vector 
with finite expectations, $m_{i,j} := \EE\bigl(Y^{(i)}_j\bigr)$, $i,j=1,\ldots, K$, and 
with joint distribution 
\[
p^{(i)}(\br) := \PP\left(Y^{(i)}_{j}=r_j, j=1,\ldots, K\right),
\]
and 
generating function 
\begin{align*}
\boldf(\bs) := (f^{(1)}(\bs),\ldots, f^{(K)}(\bs)),
\end{align*}
where $\bs= (s_1,\ldots, s_K)\in [0,1]^K$, and
\begin{align*}
f^{(i)}(\bs) =\EE\left(\prod_{j=1}^K s_j^{Y^{(i)}_j}\right)
= \sum_{\br \in \mathbb{N}^K} p^{(i)}(\br)s_1^{r_1}\cdots s_K^{r_K},\qquad
i=1,\ldots,K.
\end{align*}
All particles live and produce independently from each other, and of the past.

Let $\psi^{(i)}(t)$ denote the number of type $i$ particles at time $t$, and set
\[
\bPsi(t) := (\psi^{(1)}(t),\ldots, \psi^{(K)}(t)),\qquad t\ge0.
\]
We assume that $\bPsi(0)$ is a $K$-dimensional random vector with distribution 
$p_0(\br)$, $\br\in \mathbb{N}^K$.
Following Chapter V.7.2 in \cite{athreya:ney:1972}, 
we define a matrix 
$A:=(a_{i,j})_{i,j=1,\ldots,K}$ with
\begin{align}\label{eq:matrixA}
a_{i,j} = a_j(m_{i,j}-\ind_{\{i=j\}}).
\end{align}
Suppose that all entries of $e^{t_0 A}$ are positive for some $t_0>0$, then 
by the Perron--Frobenius theorem (cf. \cite[Theorem V.2.1]{athreya:ney:1972}),
the matrix $A$ has a largest positive eigenvalue, $\lambda_1>0$, with multiplicity 1.
Let $\bv,\bu$ be the left and right eigenvectors of $\lambda_1$ respectively, with all coordinates strictly positive, and $\bu^T\ind =1$, $\bu^T\bv =1$. Then Theorem~2 in \cite[Chapter~V.7.5]{athreya:ney:1972} gives that there exists some non-negative random variable $W$ such that
\begin{align}\label{eq:conv_bxi}
e^{-\lambda_1t}\bPsi(t)\convas W\bv.
\end{align}
Also, $\PP(W>0)>0$ if $\EE\bigl(Y^{(i)}_j\log Y^{(i)}_j\bigr)<\infty$ for all $i,j=1,2,\ldots,K$.

In \cite{rabehasaina:2021}, the multitype branching process defined above is extended to \emph{a multitype branching process with immigration} (mBI process), which is similar to the construction of a birth process with immigration; see \cite{tavare:1987,wang:resnick:2018} for details.
Suppose that $N_\theta (t)$ is the counting function of homogeneous
Poisson points $0<\tau_1<\tau_2<\ldots$ with rate
$\theta>0$. Independent of this Poisson process, assume that we have
iid
copies of a multitype branching process $\{\bPsi_k(t):t\ge 0\}_{k\ge 0}$, with the same branching mechanism as the $\bPsi(\cdot)$ process defined in the previous paragraph. 
Assume also that $\bPsi_k(0)$ is a $K$-dimensional random vector with distribution $p_0(\br)$, $\br\in \mathbb{N}^K$.
Set $\tau_0=0$, and
at times $\tau_k, k\ge 1$, new particles (i.e. immigrants) arrive.
The vector process giving the number of each type at time $t$ is:
\begin{align}
\label{eq:mBI_dist}
\bxi_\theta(t) := \sum_{k=0}^\infty
\bPsi_k(t-\tau_k)\ind_{\{t\ge\tau_k\}}=\sum_{k=0}^{N_\theta(t) } \bPsi_k(t-\tau_k),
\end{align}
and we refer to the Poisson rate parameter $\theta$ as the {immigration parameter}.
Here the $\bxi_\theta(\cdot)$ is a Markov chain on $\mathbb{N}^2$.
\begin{Theorem}\label{thm:mBI}
Let $\bxi_\theta(\cdot)$ be a mBI process as given in \eqref{eq:mBI_dist}, then as $t\to\infty$, we have
\begin{align}\label{eq:conv_bxiImmi}
e^{-\lambda_1t}\bxi_\theta(t)\convas \sum_{k = 0}^\infty e^{-\lambda_1 \tau_k} W_k\bv,
\end{align}
where $\{W_k: k\ge 0\}$ are iid, satisfy $e^{-\lambda_1 t}\bPsi_k(t)\convas W_k \bv$, and are independent from $\{\tau_k:k\ge 0\}$.
\end{Theorem}

By the independence between $\{W_k: k\ge 0\}$ and $\{\tau_k: k\ge 0\}$, we see that
\begin{align*}
\EE\left(\sum_{k = 0}^\infty e^{-\lambda_1 \tau_k} W_k\right)
&= \EE(W_0)\EE\left(\sum_{k = 0}^\infty e^{-\lambda_1 \tau_k}\right)
= \EE(W_0)\sum_{k = 0}^\infty \left(\frac{\theta}{\theta+\lambda_1}\right)^k <\infty,
\end{align*}
which implies the infinite series on the right hand side of \eqref{eq:conv_bxiImmi} converges a.s..
Note also that Theorem 3 in \cite{rabehasaina:2021} specifies the convergence in distribution in $\RR_+^K$ for the mBI process in \eqref{eq:conv_bxiImmi}, which is weaker than our results in Theorem~\ref{thm:mBI}.
The proof of Theorem~\ref{thm:mBI} is deferred to Section~\ref{sec:append_thm}.
In what follows, we will consider the $K=2$ case, and embed the in- and out-degree sequences into a sequence of two-type branching processes with immigration.

\subsubsection{Embedding}\label{subsubsec:embed}
We now introduce the embedding framework for the PA model with reciprocity, using 
mBI processes.
Let $\{\bxi_{k,\delta}(t) \equiv (\xi^{(1)}_{k,\delta}(t),\xi^{(2)}_{k,\delta}(t)): t\ge0\}_{k\ge 1}$ be a sequence of independent two-type mBI processes with life time parameters $a_1=\alpha, a_2=\gamma$,  immigration parameter $\delta\ge 0$, and offspring generating functions
\begin{align}
f^{(1)}(\bs) &= (1-\rho)s_1^2 + \rho s_1^2s_2 ,\label{eq:pgf1}\\
f^{(2)}(\bs) &= (1-\rho)s_2^2 + \rho s_1 s_2^2 ,\label{eq:pgf2}
\end{align}
for $\bs=(s_1,s_2)\in [0,1]^2$.
The initial value of every immigrant is a 2-dimensional random vector with distribution
\begin{align}
\label{eq:p0r}
p_0(\br) \equiv \left(\alpha(1-\rho)\right)^{\ind_{\{\br = (1,0)\}}}\left(\gamma(1-\rho)\right)^{\ind_{\{\br = (0,1)\}}}\rho^{\ind_{\{\br = (1,1)\}}}.
\end{align}
Initial values of $\{\bxi_{k,\delta}(t): t\ge0\}_{k\ge 1}$ will be specified during the construction process. 
Equation~\eqref{eq:pgf1} gives that at the end of the life time of a type 1 particle, with 
probability $1-\rho$, it will split into two type 1 particles, increasing the total number of
type 1 particles by 1.
With probability $\rho$, a type 1 particle will give birth to 2 type 1 particles and 1 type 2 particle upon its death, which increases the total numbers of type 1 and 2 particles both by 1. 
Similar interpretations also apply to \eqref{eq:pgf2}.
Later in Theorem~\ref{thm:embed_mBI}, we will mimic the evolution of in- and out-degree processes
of node $k$ using $\xi^{(1)}_{k,\delta}(\cdot)$ and $\xi^{(2)}_{k,\delta}(\cdot)$, respectively.

Conditioning on the current state $\bx\equiv (x_1,x_2)$,
the jump probabilities from $\bx$ to $\by$ is
\begin{align*}
P(\bx, \by) &= 
\begin{cases}
\frac{\alpha(1-\rho)(x_1+\delta)}{\alpha x_1+ \gamma x_2+\delta}, &\qquad \by=\bx+(1,0),\\
\frac{\gamma(1-\rho)(x_2+\delta)}{\alpha x_1+ \gamma x_2+\delta}, &\qquad \by=\bx+(0,1),\\
\frac{\alpha\rho(x_1+\delta)+\gamma\rho(x_2+\delta)}{\alpha x_1+ \gamma x_2+\delta}, &\qquad \by=\bx+(1,1).
\end{cases}
\end{align*}
Then we see that 
$\rho$ is the probability that when the first component jumps, it is due to both components making a jump.
Suppose that $\bxi_{1,\delta}(0)=(1,1)$, then we set $T_0=0$ and initiate $\bxi_{1,\delta}(\cdot)$ at $t=0$. Let $T_1$ be the first time when the $\bxi_{1,\delta}(\cdot)$ process jumps.
Then for $t\ge 0$, 
\begin{align*}
\PP\left(T_1>t\right)
&= \exp\left\{-t\left(\alpha\left(\bxi^{(1)}_{1,\delta}(0)+\delta\right)+\gamma\left(\bxi^{(1)}_{1,\delta}(0)+\delta\right)\right)\right\}
=  e^{-t(1+\delta)}.
\end{align*}
Hence, $T_1$ follows an exponential distribution with rate $1+\delta$.
We then initiate the process, $\{\bxi_{2,\delta}(t-T_1):t\ge T_1\}$, with one of the three initial values below:
\begin{enumerate}
\item[(i)] If the $\bxi_{1,\delta}(\cdot)$ is increased by $(1,0)$, set $\bxi_{2,\delta}(0)=(0,1)$.
\item[(ii)] If the $\bxi_{1,\delta}(\cdot)$ is increased by $(0,1)$, set $\bxi_{2,\delta}(0)=(1,0)$.
\item[(iii)] If the $\bxi_{1,\delta}(\cdot)$ is increased by $(1,1)$, set $\bxi_{2,\delta}(0)=(1,1)$.
\end{enumerate}
Therefore, $\bxi_{2,\delta}(0)$ is a 2-dimensional random vector with generating function
\begin{align}\label{eq:pgf_xi0}
\EE\left(s_1^{\xi^{(1)}_{2,\delta}(0)}s_2^{\xi^{(1)}_{2,\delta}(0)}\right)
= \alpha(1-\rho)s_2+\gamma(1-\rho)s_1+\rho s_1s_2,\qquad s_1,s_2\in [0,1].
\end{align}
Define $R_1:= \ind_{\{\bxi_{1,\delta}(T_1)=\bxi_{1,\delta}(0)+(1,1)\}}$, then $\PP(R_1=1)=\rho=1-\PP(R_1=0)$.
Set also 
$$
\mathcal{F}_{T_1}:= \sigma\left(R_1; \left\{\bxi_{k,\delta}(t-T_{k-1}):t\in [T_{k-1},T_1]\right\}_{k=1,2}\right).
$$

Next, let $T_2$ be the first time when one of the $\bxi_{k,\delta}(\cdot)$, $k=1,2$, jumps, and we
let $J_2$ be the label, 1 or 2, of the process $\bxi_{1,\delta}$, $\bxi_{2,\delta}$ that jumps first.
When the $\bxi_{J_2,\delta}(\cdot)$
process jumps and is augmented by $(1,0)$, then we initiate $\{\bxi_{3,\delta}(t-T_2):t\ge T_2\}$ with $\bxi_{3,\delta}(0)=(0,1)$. When the $\bxi_{J_2,\delta}(\cdot)$
process is augmented by $(0,1)$,  we set $\bxi_{3,\delta}(0)=(1,0)$, and
set $\bxi_{3,\delta}(0)=(1,1)$ if the $\bxi_{J_2,\delta}(\cdot)$
process is augmented by $(1,1)$.
Define also that
\[
R_2 = \ind_{\left\{\bxi_{J_2,\delta}(T_2-T_{J_2-1})=\bxi_{J_2,\delta}(T_1-T_{J_2-1})+(1,1)\right\}},
\]
and write $\PP^{\mathcal{F}_{T_1}}(\cdot):= \PP(\cdot|\mathcal{F}_{T_1})$.
Then we have
\begin{align*}
&\PP^{\mathcal{F}_{T_1}}\left(R_2=1, J_2=1, T_2-T_1>t\right) \\
&= \rho\times\frac{\alpha  \left(\xi^{(1)}_{1,\delta}(T_1)+\delta
\right)+\gamma\left(\xi^{(2)}_{1,\delta}(T_1)+\delta
\right)}{\alpha\sum_{k=1}^2 \xi^{(1)}_{k,\delta}(T_1-T_{k-1})
+ \gamma\sum_{k=1}^2 \xi^{(2)}_{k,\delta}(T_1-T_{k-1}) +2\delta},\\
\intertext{and since $\sum_{k=1}^2 \xi^{(1)}_{k,\delta}(T_1-T_{k-1})=\sum_{k=1}^2 \xi^{(2)}_{k,\delta}(T_1-T_{k-1})= 2+R_1$, we have}
&= \rho\times\frac{\alpha  \left(\xi^{(1)}_{1,\delta}(T_1)+\delta
\right)+\gamma\left(\xi^{(2)}_{1,\delta}(T_1)+\delta
\right)}{2(1+\delta)+R_1} e^{-t(2(1+\delta)+R_1)}\\
&= \PP^{\mathcal{F}_{T_1}}\left(R_2=1\right)\PP^{\mathcal{F}_{T_1}}\left(J_2=1\right)
\PP^{\mathcal{F}_{T_1}}\left(T_2-T_1>t\right).
\end{align*}
Therefore, $(R_2,J_2, T_2-T_1)$ are conditionally independent under $\PP^{\mathcal{F}_{T_1}}$.

In general, for $n\ge 1$, suppose that we have initiated $n+1$ mBI processes at time $T_n$,
\begin{align}\label{eq:mBIs}
\{\bxi_{k,\delta}(t-T_{k-1}): t\ge T_{k-1}\}_{1\le k\le n+1}.
\end{align}
Define $T_{n+1}$ as the first time when one of the processes in \eqref{eq:mBIs} jumps, and let
 $J_{n+1}$ be the label of the process that jumps at $T_{n+1}$. If the $\bxi_{J_{n+1},\delta}(\cdot)$
process is augmented by $(1,0)$, then we initiate $\{\bxi_{n+2,\delta}(t-T_{n+1}):t\ge T_{n+1}\}$ with $\bxi_{n+2,\delta}(0)=(0,1)$. If the $\bxi_{J_{n+1},\delta}(\cdot)$
process is augmented by $(0,1)$, set $\bxi_{n+2,\delta}(0)=(1,0)$, and
set $\bxi_{n+2,\delta}(0)=(1,1)$ if the $\bxi_{J_{n+1},\delta}(\cdot)$
process is augmented by $(1,1)$.
Hence, $\bxi_{n+2,\delta}(0)$ has the same generating function as in \eqref{eq:pgf_xi0}.
Write 
\[
R_{n+1} = \ind_{\left\{\bxi_{J_{n+1},\delta}(T_{n+1}-T_{J_{n+1}-1})=\bxi_{J_{n+1},\delta}(T_n-T_{J_{n+1}-1})+(1,1)\right\}},
\]
and define
\[
\mathcal{F}_{T_n}:= \sigma\left(\{R_l\}_{l=1}^n; \left\{\bxi_{k,\delta}(t-T_{k-1}):t\in [T_{k-1},T_n]\right\}_{1\le k\le n+1}\right).
\]
Using the fact that
$\sum_{k=1}^{n+1}\xi^{(1)}_{k,\delta}(T_{n}-T_{k-1}) = \sum_{k=1}^{n+1}\xi^{(2)}_{k,\delta}(T_{n}-T_{k-1})= n+1+\sum_{l=1}^n R_l$, 
 we have
\begin{align}
&\PP^{\mathcal{F}_{T_n}}\left(R_{n+1}=1, J_{n+1}=w, T_{n+1}-T_n>t\right) \nonumber\\
&= \rho\times\frac{\alpha  \left(\xi^{(1)}_{w,\delta}(T_n-T_{w-1})+\delta
\right)+\gamma\left(\xi^{(2)}_{w,\delta}(T_n-T_{w-1})+\delta
\right)}{\alpha\sum_{k=1}^{n+1} \xi^{(1)}_{k,\delta}(T_{n}-T_{k-1})
+ \gamma\sum_{k=1}^{n-1} \xi^{(2)}_{k,\delta}(T_n-T_{k-1}) +(n+1)\delta}
\nonumber\\
&\quad \times \exp\left\{-t\left(\alpha\sum_{k=1}^{n+1} \xi^{(1)}_{k,\delta}(T_{n}-T_{k-1})
+ \gamma\sum_{k=1}^{n-1} \xi^{(2)}_{k,\delta}(T_n-T_{k-1}) +(n+1)\delta\right)\right\}\nonumber\\
&= \rho e^{-t\left((n+1)(1+\delta)+\sum_{l=1}^n R_l\right)}\nonumber\\
&\quad \times\frac{\alpha  \left(\xi^{(1)}_{w,\delta}(T_{n}-T_{w-1})+\delta
\right)+\gamma \left(\xi^{(2)}_{w,\delta}(T_{n}-T_{w-1})+\delta
\right)}{(n+1)(1+\delta)+\sum_{l=1}^n R_l} \nonumber\\
&= \PP^{\mathcal{F}_{T_n}}\left(R_{n+1}=1\right)\PP^{\mathcal{F}_{T_n}}\left(J_{n+1}=w\right)
\PP^{\mathcal{F}_{T_n}}\left(T_{n+1}-T_n>t\right), \label{eq:JnRn}
\end{align}
i.e. $(J_{n+1},R_{n+1}, T_{n+1}-T_n)$ are conditionally independent under $\PP^{\mathcal{F}_{T_n}}$.

Write 
\[
{\bxi}^*_\delta(T_n):=\left(\bxi_{1,\delta}(T_n),\bxi_{2,\delta}(T_n-T_1),\ldots, \bxi_{n+1,\delta}(0),
(0,0),\ldots\right), \qquad n\ge 0,
\]
then the embedding framework described above shows that
$\{{\bxi}^*_\delta(T_n): n\ge 0\}$ is Markovian on $\left(\mathbb{N}^2\right)^\infty$.
In the following theorem, we embed the evolution of in- and out-degree processes into the mBI  framework.
\begin{Theorem}\label{thm:embed_mBI}
In $\left(\mathbb{N}^2\right)^\infty$, define the in- and out-degree sequences as
\[
\bD(n) := \left(\bigl(\Din_1(n),\Dout_1(n)\bigr),\ldots, \bigl(\Din_{n+1}(n),\Dout_{n+1}(n)\bigr),
(0,0),\ldots\right).
\]
Then for $\{T_k: k\ge 0\}$ and $\{\bxi_{k,\delta}(t-T_{k-1}): t\ge T_{k-1}\}_{k\ge 1}$ constructed above,
we have that in $\left(\left(\mathbb{N}^2\right)^\infty\right)^\infty$,
\begin{align*}
\bigl\{\bD(n): n\ge 0\bigr\} \stackrel{d}{=} \left\{{\bxi}^*_\delta(T_n): n\ge 0\right\}.
\end{align*}
\end{Theorem}
\begin{proof}
By the model description in Section~\ref{sec:model}, $\{\bD(n):n\ge 0\}$ is Markovian on $\left(\mathbb{N}^2\right)^\infty$, so
it suffices to check the transition probability from $\bD(n)$ to $\bD(n+1)$ agrees with that from 
${\bxi}^*_\delta(T_n)$ to ${\bxi}^*_\delta(T_{n+1})$.
Write
\begin{align*}
\bolde_w^{(1)} &:= \left(\bigl(0,0\bigr),\ldots,\bigl(0,0\bigr),\underbrace{\bigl(1,0\bigr)}_{\text{$w$-th entry}},\bigl(0,0\bigr),\ldots\right),\\
\bolde_w^{(2)} &:= \left(\bigl(0,0\bigr),\ldots,\bigl(0,0\bigr),\underbrace{\bigl(0,1\bigr)}_{\text{$w$-th entry}},\bigl(0,0\bigr),\ldots\right),
\end{align*}
and let $\mathcal{G}_n$ denote the $\sigma$-field generated by the history of the network up to $n$ steps.
Then we have
\begin{align}
\PP^{\mathcal{G}_n}\left(\bD(n+1)=\bD(n)+\bolde_w^{(1)}\right)
&= \frac{\alpha(1-\rho)(\Din_w(n)+\delta)}{|E(n)|+\delta (n+1)},\label{eq:PAtrans1}\\
\PP^{\mathcal{G}_n}\left(\bD(n+1)=\bD(n)+\bolde_w^{(2)}\right)
&= \frac{\gamma(1-\rho)(\Dout_w(n)+\delta)}{|E(n)|+\delta (n+1)},\label{eq:PAtrans2}\\
\PP^{\mathcal{G}_n}\left(\bD(n+1)=\bD(n)+\bolde_w^{(1)}+\bolde_w^{(2)}\right)
&= \rho\frac{\alpha(\Din_w(n)+\delta)+\gamma(\Dout_w(n)+\delta)}{|E(n)|+\delta (n+1)},
\label{eq:PAtrans3}
\end{align}
where $|E(n)|-(n+1)$ follows a binomial distribution with size $n$ and success probability $\rho$.

By \eqref{eq:JnRn}, we see that
\begin{align}
\PP^{\mathcal{F}_{T_n}}& ({\bxi}^*_\delta(T_{n+1}) = {\bxi}^*_\delta(T_n)+\bolde_w^{(1)}
+\bolde_w^{(2)})
= \PP^{\mathcal{F}_{T_n}} (J_{n+1}=w, R_{n+1}=1)\nonumber\\
&= \frac{\alpha \rho \left(\xi^{(1)}_{w,\delta}(T_{n}-T_{w-1})+\delta
\right)+\gamma\rho \left(\xi^{(2)}_{w,\delta}(T_{n}-T_{w-1})+\delta
\right)}{(n+1)(1+\delta)+\sum_{l=1}^n R_l}.\label{eq:mBItrans3}
\end{align}
Note also that for $n\ge 1$,
$\PP^{\mathcal{F}_{T_{n-1}}}(R_n = 1) = \rho = 1-\PP^{\mathcal{F}_{T_{n-1}}}(R_n = 0)$.
Then for $n\ge 1$ and $r_l\in \{0,1\}$, $1\le l\le n$, we have
\begin{align*}
\PP(R_l=r_l, 1\le l\le n) &= \EE\left[\ind_{\left\{R_l=r_l, 1\le l\le n-1\right\}}
\PP^{\mathcal{F}_{T_{n-1}}}(R_n = r_n)\right],\\
\intertext{and if $r_n=1$, we have}
&= \PP(R_l=r_l, 1\le l\le n-1) \PP(R_n=1)\\
&=\cdots =\prod_{l=1}^n \PP(R_l=r_l).
\end{align*}
Therefore, $\{R_l: 1\le l\le n\}$ are iid Bernoulli random variables with $\PP(R_1=1)=\rho$, then 
$\sum_{l=1}^n R_l$ has the same distribution as $|E(n)|-(n+1)$, which
shows that the transition probability in \eqref{eq:mBItrans3} agrees with that in \eqref{eq:PAtrans3}.
Similarly, 
\begin{align*}
\PP^{\mathcal{F}_{T_n}} \left({\bxi}^*_\delta(T_{n+1}) = {\bxi}^*_\delta(T_n)+\bolde_w^{(1)}\right)
& = \frac{\alpha (1-\rho) (\xi_{w,\delta}^{(1)}(T_n-T_{w-1})+\delta)}{(n+1)(1+\delta)+\sum_{l=1}^n R_l},\\
\PP^{\mathcal{F}_{T_n}} \left({\bxi}^*_\delta(T_{n+1}) = {\bxi}^*_\delta(T_n)+\bolde_w^{(2)}\right)
& = \frac{\gamma (1-\rho) (\xi_{w,\delta}^{(2)}(T_n-T_{w-1})+\delta)}{(n+1)(1+\delta)+\sum_{l=1}^n R_l},
\end{align*}
which agree with transition probabilities in \eqref{eq:PAtrans1} and \eqref{eq:PAtrans2}, respectively.
\end{proof}

\subsection{Multivariate regular variation}
Later in Section~\ref{subsec:Nij}, we will derive the asymptotic distribution of the joint in- and out-degree counts, and show that they are jointly heavy tailed. To formalize our analysis, we provide some useful definitions related to \emph{multivariate regular variation} (MRV).

Suppose that $\mathbb{C}_0\subset\mathbb{C}\subset\mathbb{R}_+^2$ are two closed cones, and we provide the definition of $\mathbb{M}$-convergence in Definition~\ref{def:Mconv}
(cf. \cite{lindskog:resnick:roy:2014,hult:lindskog:2006a,das:mitra:resnick:2013,kulik:soulier:2020,basrak:planinic:2019}) on $\mathbb{C}\setminus \mathbb{C}_0$, which lays the theoretical foundation of regularly varying measures (cf. Definition~\ref{def:MRV}).
\begin{Definition}\label{def:Mconv}
Let $\mathbb{M}(\mathbb{C}\setminus \mathbb{C}_0)$ be the set of Borel
measures on $\mathbb{C}\setminus \mathbb{C}_0$ which are finite on
sets bounded away from $\mathbb{C}_0$, and
$\mathcal{C}(\mathbb{C}\setminus \mathbb{C}_0)$ be the set of
continuous, bounded, non-negative functions on $\mathbb{C}\setminus
\mathbb{C}_0$ whose supports are bounded away from  $\mathbb{C}_0$. Then for $\mu_n,\mu \in \mathbb{M}(\mathbb{C}\setminus
\mathbb{C}_0)$, we say $\mu_n \to \mu$ in
$\mathbb{M}(\mathbb{C}\setminus \mathbb{C}_0)$, if $\int f\dd
\mu_n\to\int f\dd \mu$ for all $f\in \mathcal{C}(\mathbb{C}\setminus
\mathbb{C}_0)$. 
\end{Definition} 

Without loss of generality \cite{lindskog:resnick:roy:2014}, we can
and do take functions in  $\mathcal{C}(\mathbb{C}\setminus
\mathbb{C}_0)$ to be uniformly continuous as well.  Denote the modulus
of continuity of a uniformly continuous function $f:\R_+^p \mapsto
\R_+$ by
\begin{equation}\label{e:defModCon}
\Delta_f(\delta)=\sup\{ |f(\bx)-f(\by)|:
d(\bx,\by)<\delta\}\end{equation}
 where $d(\cdot,\cdot)$ is an appropriate metric
on the domain of $f$. Uniform continuity means $\lim_{\delta \to 0}
\Delta_f(\delta)=0.$ 
 
\begin{Definition}\label{def:MRV}
The distribution of a
 random vector $\bZ$ on $\mathbb{R}_+^2$, i.e.
$\PP(\bZ\in\cdot)$, 
  is (standard) regularly varying on $\mathbb{C}\setminus \mathbb{C}_0$ with index $c>0$ (written as $\bZ\in \text{MRV}(c, b(t), \nu, \mathbb{C}\setminus \mathbb{C}_0)$) if there exists some scaling function $b(t)\in \text{RV}_{1/c}$ and a limit measure $\nu(\cdot)\in \mathbb{M}(\mathbb{C}\setminus \mathbb{C}_0)$ such that
as $t\to\infty$,
\begin{equation}\label{eq:def_mrv}
t\PP(\bZ/b(t)\in\cdot)\rightarrow \nu(\cdot),\qquad\text{in }\mathbb{M}(\mathbb{C}\setminus \mathbb{C}_0).
\end{equation}
\end{Definition}

When analyzing the asymptotic dependence between components of a
bivariate random vector
$\bZ$ satisfying \eqref{eq:def_mrv},  it is often informative  
to make a polar coordinate transform and
consider the transformed points located on the $L_1$ unit sphere
\begin{align}
\label{eq:map_L1}
(x,y)\mapsto\left(\frac{x}{|x|+|y|},\frac{y}{|x|+|y|}\right),
\end{align}
after thresholding the data according to the
$L_1$ norm.
The plot of the transformed points is referred to as the diamond plot,
which provides a visualization of dependence. 
Also, provided that $|x|+|y|$ is larger than some predetermined threshold, the density plot of thresholded values $x/(|x|+|y|)$ is called the angular density plot.
These plots characterize the asymptotic dependence structure for extremal observations.

In Section~\ref{subsec:HRV}, we apply the transformation in \eqref{eq:map_L1} to nodes with large
in- and out-degrees, and find that the angular density plot concentrates around 
some particular value. In the terminology of \cite{das:resnick:2017}, this indicates that the limiting in- and out-degree pair has \emph{full asymptotic dependence}.

\subsubsection{A modification of Breiman's
  Theorem.}\label{subsub:breiman} For the study of the regular
variation properties of the asymptotic distribution of degree
frequencies given in Theorem \ref{thm:IO_mrv}, we need the following
generalization of Breiman's Theorem \cite{breiman:1965}. This result
about products has spawned many proofs and generalizations. See for
instance
\cite{resnickbook:2007,
  kulik:soulier:2020,
  fougeres:mercadier:2012,
  maulik:resnick:rootzen:2002, chen:chen:gao:2019,
  basrak:davis:mikosch:2002b}.

\begin{Theorem}\label{th:extendBrei}
Suppose $\{\bxi(t): t\geq 0\}$ is an $\R_+^p$-valued stochastic
process 
for some $p\geq 1$.  Let $X$ be a positive random variable with
regularly varying distribution satisfying for some scaling function $b(t)$,
$$\lim_{t\to\infty} t\PP( X/b(t) >x) =x^{-c} =:\nu_c \bigl((x,\infty)\bigr), \quad x>0, c>0.$$
Further suppose
\begin{enumerate}
\item For some finite and positive random vector $\bxi_\infty$,
  $$\lim_{ t \to \infty} {\bxi(t)} =\bxi_\infty \quad (\text{almost surely});$$
  \item The random variable $X$ and the process $\bxi (\cdot)$ are independent.
  \end{enumerate}
  Then:

  (i) In $\mathbb{M}(\R_+^p \times (\R_+\setminus \{0\}))$,
  \begin{equation}\label{e:beforeMult}
    t\PP\Bigl[ \Bigl({\bxi(X)}, \frac{X}{b(t)}\Bigr) \in \cdot \,\Bigr]
    \longrightarrow \PP(\bxi_\infty \in \cdot \,) \times
    \nu_c (\cdot)=:\eta(\cdot).\end{equation}
  If $\bxi_\infty$ is of the form $\bxi_\infty =:L\bv$ where $L>0$
  almost surely and $\bv \in (0,\infty)^p$, then $\eta (\cdot)$
  concentrates on the subcone $\mathcal{L}\times (\R_+\setminus
  \{0\})$ where $\mathcal{L}=\{\theta \bv : \theta>0\}$.

  (ii) If additionally, for some $c'>c$ we have the condition
  \begin{equation}\label{e:extraCond}
  \kappa:=  \sup_{t\geq 0}   \EE\left[ \Bigl(  { \|\bxi(t) \|}  \Bigr)^{c'}\right]
    <\infty,
  \end{equation}
  for some $L_p$ norm $\|\cdot\|$,  then the product 
  of components in
  \eqref{e:beforeMult}, $\bxi(X)X $, has a regularly varying distribution with
  scaling function $b(t)$ and 
    in $\mathbb{M}(\R_+^p \setminus
\{\bzero\})$, 
  \begin{equation}\label{e:prodOK}
    t\PP\Bigl[\frac{X\bxi(X)}{b(t)} \in \cdot \,\Bigr]
    \longrightarrow \left(\PP(\bxi_\infty \in \cdot \,) \times \nu_c\right) \circ h^{-1},
  \end{equation}
where $h(\by,x)=x\by$.
\end{Theorem}

For the classical Breiman Theorem where $p=1$ and $\bxi(t)\equiv \bxi_\infty$,
 \eqref{e:extraCond} is the expected moment condition.
Detailed proofs of Theorem~\ref{th:extendBrei} are collected in Section~\ref{subsec:pf_breim}.


\section{Asymptotic Results}\label{sec:res}
With the embedding framework in Theorem~\ref{thm:embed_mBI}, we derive
theoretical results on: (1) the a.s. growth of in- and out-degrees for
a fixed node; (2) the limiting joint distribution of in- and
out-degrees based on degree counts; (3) the asymptotic dependence
structure for large in- and out-degrees.

\subsection{Convergence of Degrees for a Fixed Node}
By the construction of the $\{\bxi_{k,\delta}(t):t\ge 0\}_{k\ge 1}$, we see that the corresponding $A$ matrix as introduced in \eqref{eq:matrixA} is equal to
\begin{equation}\label{eq:defA}
A=
\begin{bmatrix}
\alpha & \alpha\rho\\
\gamma\rho & \gamma
\end{bmatrix},
\end{equation}
with the largest eigenvalue 
$$
\lambda_1= \frac{1}{2}\left(1+\sqrt{(\alpha-\gamma)^2+4\alpha\gamma\rho^2}\right)
=: \frac{1+\sqrt{\Delta}}{2}.
$$
The left and right eigenvectors with respect to $\lambda_1$ are
\begin{align}
\bv &= \left(\frac{\alpha+\gamma(2\rho-1)+\sqrt{\Delta}}{2\sqrt{\Delta}}\right)\begin{bmatrix}
1\medskip\\
\frac{\gamma-\alpha+\sqrt{\Delta}}{2\gamma\rho}
\end{bmatrix}, \label{eq:defv}\\
\bu &= \left(\frac{\alpha+\gamma(2\rho-1)+\sqrt{\Delta}}{2\sqrt{\Delta}}\right)^{-1}\begin{bmatrix}
\frac{\alpha-\gamma +\sqrt{\Delta}}{2\sqrt{\Delta}}
\medskip\\
\frac{\gamma \rho}{\sqrt{\Delta}}
\end{bmatrix} , \label{eq:defu}
\end{align}
which satisfy $\bu^T \bv=1$ and $\bu^T\ind = 1$.
The following theorem gives the joint convergence of $(\Din_w(n), \Dout_w(n))$ for a fixed node $w$.

\begin{Theorem}\label{thm:fixed_node}
Suppose that $\bv,\bu\in \mathbb{R}_+^2$ are as defined in \eqref{eq:defv} and \eqref{eq:defu}. 
Then for $w\ge 1$, there exists some finite random variable $L_w$ such that
\begin{align}
\label{eq:conv_deg}
\left(\frac{\Din_w(n)}{n^{\lambda_1/(1+\rho+\delta)}},\,\frac{\Dout_w(n)}{n^{\lambda_1/(1+\rho+\delta)}}\right)\convas
  L_w \bv.
\end{align}
\end{Theorem}
\begin{proof}
By the embedding results in Theorem~\ref{thm:embed_mBI}, we prove  
\eqref{eq:conv_deg} using the mBI framework, i.e. by proving
\begin{align*}
\frac{\bxi_{w,\delta}(T_n-{T_{w-1}})}{n^{\lambda_1}}\convas
  L_w \bv.
\end{align*}
Based on the definition of $\{T_k:k\ge 0\}$,
 we apply \cite[Theorem III.9.1]{athreya:ney:1972} to obtain that
\[
T_n -\sum_{k=1}^n \frac{1}{(1+\delta)k+\sum_{l=1}^k R_l}
\] 
is an $L_2$-bounded martingale with respect to $\{\mathcal{F}_{T_k}:k\ge 0\}$.
 Then applying the proof machinery in \cite[Proposition 2.2]{athreya:ghosh:sethuraman:2008}
 gives that
 there exists some finite random variable $Z$ such that
 \begin{align}\label{eq:conv_Tn}
 T_n - \frac{1}{1+\rho+\delta} \log n \convas  Z,
 \end{align}
 Then we are left to show 
 the convergence of $e^{-\lambda_1 (T_n-T_{w-1})} \bxi_{w,\delta}(T_n-T_{w-1})$ for $w\ge 1$, as $n\to\infty$.

Recall the construction of $\{\bxi_{k,\delta}(\cdot)\}_{k\ge 1}$ in Section~\ref{subsubsec:embed}. 
For $w=1$, the $\bxi_{1,\delta}(0)=(1,1)$, and the offspring generating functions are given in \eqref{eq:pgf1} and \eqref{eq:pgf2}. 
This leads to a slightly different mBI setup from the one given in Section~\ref{subsec:mBI}, which we now describe.
Let $\{\bzeta_0(t):t\ge 0\}$
be a two-type branching process with $\bzeta_0(0)=(1,1)$ and offspring distribution as in \eqref{eq:pgf1} and \eqref{eq:pgf2}. Assume $\{\bzeta_k(\cdot)\}_{k\ge 1}$ are independent from $\bzeta_0(\cdot)$, and 
are iid copies of the two-type branching process, where 
$\bzeta_k(0)$ has the same distribution as in \eqref{eq:p0r}, and $\bzeta_k(\cdot)$ has 
the same offspring generating functions as in \eqref{eq:pgf1} and \eqref{eq:pgf2}.
Denote jump times of an independent Poisson process with intensity $\delta>0$ as  $\{\tau_k\}_{k\ge 1}$, 
and we have
 \[
 \bxi_{1,\delta}(t) {=} \bzeta_0(t) + \sum_{k=1}^\infty \bzeta_k(t-\tau_k)\ind_{\{t\ge \tau_k\}}.
 \]
 Therefore, by the convergence results in \eqref{eq:conv_bxi} and \eqref{eq:conv_bxiImmi}, we have that as $t\to\infty$,
 \[
 e^{-\lambda_1 t} \bxi_{1,\delta}(t) \convas  \left(\widetilde{W}_0 + \sum_{k=1}^\infty e^{-\lambda_1\tau_k}\widetilde{W}_k\right) \bv =: L_1\bv,
 \]
 where $\{\widetilde{W}_k: k\ge 0\}$ satisfy $e^{-\lambda_1 t}\bzeta_k(t)  \widetilde{W}_k \bv$, and are independent from $\{\tau_k:k\ge 1\}$.
 Hence, as $n\to\infty$, 
 \begin{align}\label{eq:conv_bxi1}
 e^{-\lambda_1 T_n} \bxi_{1,\delta}(T_n) \convas  \left(\widetilde{W}_0 + \sum_{k=1}^\infty e^{-\lambda_1\tau_k}\widetilde{W}_k\right) \bv.
 \end{align}
 Combining \eqref{eq:conv_Tn} with \eqref{eq:conv_bxi1} gives
 \[
 \frac{\bxi_{1,\delta}(T_n)}{n^{\lambda_1/(1+\rho+\delta)}} \convas e^{\lambda_1 Z} \left(\widetilde{W}_0 + \sum_{k=1}^\infty e^{-\lambda_1\tau_k}\widetilde{W}_k\right) \bv. 
 \]
 
 For $w\ge 2$, we need to modify the initialization of the two-type branching process at $t=0$, i.e. we assume at $t=0$, for the initial process $\{\bzeta'_0(t):t\ge 0\}$, $\bzeta'_0(0)$ is a random vector with generating function as in \eqref{eq:pgf_xi0}. We assume that the offspring generating functions for $\bzeta'_0(\cdot)$ are the same as in \eqref{eq:pgf1} and \eqref{eq:pgf2}. Then by 
\eqref{eq:conv_bxi}, there exists some random variable $W_0'$ such that
 \[
 e^{-\lambda_1 t}\bzeta'_0(t)\convas W_0'\bv.
 \]
Let $\{\tau'_k\}_{k\ge 1}$ be jump times of a Poisson process with intensity $\delta>0$, independent from $\{\tau_k\}_{k\ge 1}$.
 Assume also that $\{\bzeta'_k(\cdot)\}_{k\ge 1}$ are iid copies of the $\bzeta_1(\cdot)$ process, and are independent from both $\bzeta'_0(\cdot)$ and $\{\tau'_k:k\ge 1\}$.
 Then for $w\ge 2$,
 \[
 \bxi_{w,\delta}(t) = \bzeta'_0(t) + \sum_{k=1}^\infty \bzeta'_k(t-\tau'_k)\ind_{\{t\ge \tau'_k\}},
 \]
 which leads to
 \begin{align}\label{eq:conv_bxi2}
 e^{-\lambda_1 (T_n-T_{w-1})} \bxi_{w,\delta}(T_n-T_{w-1}) \convas \left({W}_0' + \sum_{k=1}^\infty e^{-\lambda_1\tau'_k}\widetilde{W}_k\right) \bv.
 \end{align}
 Again, combining \eqref{eq:conv_Tn} with \eqref{eq:conv_bxi2} gives 
 \[
 \frac{\bxi_{w,\delta}(T_n-T_{w-1})}{n^{\lambda_1/(1+\rho+\delta)}} \convas e^{\lambda_1(Z-T_{w-1})} \left({W}'_0 + \sum_{k=1}^\infty e^{-\lambda_1\tau'_k}\widetilde{W}_k\right) \bv =: L_w\bv,
 \]
 which completes the proof of the theorem.
\end{proof}

One special case is having either $\alpha=1$ or $\gamma=1$, where we 
specify the distribution of $L_w$. We only consider the $\alpha=1$ case here, and results for $\gamma=1$ follow from a similar argument. When $\alpha=1$, the matrix $A$ associated with $\{\bxi_{w,\delta}(t): t\ge 0\}$ becomes 
\begin{equation}\label{eq:A_alpha_1}
A=
\begin{bmatrix}
1 & \rho\\
0 & 0
\end{bmatrix},
\end{equation}
with $\lambda_1=1$ and $\bv = (1,\rho)^T$. Then Theorem~\ref{thm:fixed_node} implies 
\[
\left(\frac{\Din_w(n)}{n^{1/(1+\rho+\delta)}}, \frac{\Dout_w(n)}{n^{1/(1+\rho+\delta)}}\right)
\convas L_w \begin{bmatrix}
1 \\
\rho
\end{bmatrix}.
\]
Note also that by the special structure in \eqref{eq:A_alpha_1}, we have
\[
\bxi^{(2)}_{w,\delta}(T_n-T_{w-1}) = \ind_{\{w=1\}}+\sum_{m=1}^{\xi^{(1)}_{w,\delta}(T_n-T_{w-1})} R_m,
\]
and $\{\xi^{(1)}_{w,\delta}(t): t\ge 0\}$ is a single-type birth immigration process with 
$\xi^{(1)}_{w,\delta}(0)=\ind_{\{w=1\}}+R_{w-1}\ind_{\{w\ge 2\}}$,
and transition rate $q_{h,h+1}=h+\delta$.
Therefore by \cite{tavare:1987}, when $\alpha=1$, the limiting random variable $L_w$ has pdf
\begin{align}
\label{eq:pdf_L}
f_{L_w}(x) = \frac{1-\rho}{\Gamma\left(\ind_{\{w=1\}}+\delta\right)} x^{\ind_{\{w=1\}}+\delta-1}e^{-x}
+ \frac{\rho}{\Gamma\left(1+\delta\right)} x^{\delta}e^{-x},\qquad x\ge 0.
\end{align}

\subsection{Convergence of Degree Counts}\label{subsec:Nij}
In the current section, we study the limiting behavior of joint degree counts:
\begin{align*}
N_{m,l}(n) = \sum_{w=1}^{n+1} \ind_{\left\{(\Din_w(n),\Dout_w(n))=(m,l)\right\}}.
\end{align*}
Using the embedding result, the following theorem shows the convergence
of $N_{m,l}(n)/n$.
\begin{Theorem}\label{thm:limitNij}
Suppose that $\{\widetilde{\bxi}_\delta(t):t\ge 0\}$ is a two-type mBI process with the same 
initialization and branching mechanism as $\bxi_{2,\delta}(\cdot)$
given in \eqref{eq:pgf_xi0}.  Then as $n\to\infty$, we have
for $m,l\ge 0$,
\begin{align}\label{eq:Nij}
\frac{N_{m,l}(n)}{n}&\convp \int_0^\infty (1+\rho+\delta)e^{-t(1+\rho+\delta)}
\PP\left(\widetilde{\bxi}_\delta(t)=(m,l)\right)\dd t \nonumber \\
&= : \PP\left(\big(\mathcal{I},\mathcal{O}\bigr) = (m,l)\right).
\end{align}
\end{Theorem}


\begin{proof}
By the embedding results in Theorem~\ref{thm:embed_mBI}, we have
\begin{align}\label{eq:Nij_dist}
\frac{N_{m,l}(n)}{n}\stackrel{d}{=} \frac{1}{n}\sum_{w=2}^{n+1}\ind_{\left\{\bxi_{w,\delta}(T_n-T_{w-1})=(m,l)\right\}}
+ \frac{1}{n}\ind_{\left\{\bxi_{1,\delta}(T_n)=(m,l)\right\}},
\end{align}
and we see that the second term on the right hand side goes to 0 a.s. as $n\to\infty$.
Then it suffices to consider the limiting behavior of the first term in \eqref{eq:Nij_dist},
which is divided into different parts below.
\begin{align*}
&\frac{1}{n}\sum_{w=2}^{n+1}\ind_{\left\{\bxi_{w,\delta}(T_n-T_{w-1})=(m,l)\right\}}\\
&= \left(\frac{1}{n}\sum_{w=2}^{n+1}\ind_{\left\{\bxi_{w,\delta}(T_n-T_{w-1})=(m,l)\right\}}
- \frac{1}{n}\sum_{w=2}^{n+1}\ind_{\left\{\bxi_{w,\delta}\left(\frac{\log(n/w)}{1+\rho+\delta}\right)=(m,l)\right\}}
\right)\\
&\,+ 
\left(
\frac{1}{n}\sum_{w=2}^{n+1}\ind_{\left\{\bxi_{w,\delta}\left(\frac{\log(n/w)}{1+\rho+\delta}\right)=(m,l)\right\}}- \frac{1}{n}\sum_{w=2}^{n+1}\PP\left[\bxi_{w,\delta}\left(\frac{\log(n/w)}{1+\rho+\delta}\right)=(m,l)\right]
\right)\\
&\,+ \left(
\frac{1}{n}\sum_{w=2}^{n+1}\PP\left[\bxi_{w,\delta}\left(\frac{\log(n/w)}{1+\rho+\delta}\right)=(m,l)\right]\right.\\
&\left. \qquad - \int_0^1 \PP\left[\bxi_{2,\delta}\left(-\frac{\log t}{1+\rho+\delta}\right)=(m,l)\right]\dd t
\right)\\
&\, + \int_0^1 \PP\left[\bxi_{2,\delta}\left(-\frac{\log t}{1+\rho+\delta}\right)=(m,l)\right]\dd t\\
&=: A_1(n)+A_2(n)+A_3(n)+A_4.
\end{align*}
Note that by a change of variable argument, $A_4$ is identical to the right hand side of \eqref{eq:Nij},
and we now show that $A_k(n)\convp 0$ for $k=1,2$, and $A_3(n)\to 0$. 

For $A_1(n)$, we have
\begin{align}
\EE|A_1(n)|&\le \frac{1}{n}\sum_{w=2}^{n+1}\EE\left|\ind_{\{\xi^{(1)}_{w,\delta}(T_n-T_{w-1})=m\}}
-\ind_{\{\xi^{(1)}_{w,\delta}\bigl(\log(n/w)/(1+\rho+\delta)\bigr)=m\}}\right|\nonumber\\
&\quad +\frac{1}{n}\sum_{w=2}^{n+1}\EE\left|\ind_{\{\xi^{(2)}_{w,\delta}(T_n-T_{w-1})=l\}}
-\ind_{\{\xi^{(2)}_{w,\delta}\bigl(\log(n/w)/(1+\rho+\delta)\bigr)=l\}}\right|.
\label{eq:A1_bound}
\end{align}
Since both $\xi^{(1)}_{w,\delta}(\cdot)$ and $\xi^{(2)}_{w,\delta}(\cdot)$ are non-explosive, i.e.
 have finite number of jumps in any finite time interval $[0,K]$ a.s., then 
 Lemma~3.1 in \cite{athreya:ghosh:sethuraman:2008} implies that for all $K>0$ and $w\ge 2$,
 \[
 \lim_{\epsilon\downarrow 0}\sup_{t\in [0,K]} \PP\left(\xi^{(i)}_{w,\delta}(t+\epsilon)-
 \xi^{(i)}_{w,\delta}\bigl((t-\epsilon)\wedge 0)\bigr)\ge 1\right)=0,\qquad i=1,2.
 \]
 Also, by \cite[Corollary 2.1(iii)]{athreya:ghosh:sethuraman:2008}, we see that for $\eta>0$,
 \[
 \sup_{n\eta \le w\le n}\left|T_n-T_{w-1}-\frac{1}{1+\rho+\delta}\log(n/w)\right|\convas 0.
 \]
 Then using the proof machinery for \cite[Theorem 1.2, pp 489--490]{athreya:ghosh:sethuraman:2008}, we have
 \begin{align*}
 &\left|\ind_{\{\xi^{(1)}_{w,\delta}(T_n-T_{w-1})=m\}}
-\ind_{\{\xi^{(1)}_{w,\delta}\bigl(\log(n/w)/(1+\rho+\delta)\bigr)=l\}}\right|\\
&\le \sup_{t\in [0,-\log\eta/(1+\rho+\delta)]}\PP\left(\xi^{(1)}_{2,\delta}(t+\epsilon)-
 \xi^{(1)}_{2,\delta}\bigl((t-\epsilon)\wedge 0)\bigr)\ge 1\right)\\
&\quad +\PP\left(\sup_{n\eta \le w\le n}\left|T_n-T_{w-1}-\frac{1}{1+\rho+\delta}\log(n/w)\right|\ge \epsilon\right)=: p_1(\epsilon,\eta).
 \end{align*}
 Similarly,
 \begin{align*}
 &\left|\ind_{\{\xi^{(2)}_{w,\delta}(T_n-T_{w-1})=l\}}
-\ind_{\{\xi^{(2)}_{w,\delta}\bigl(\log(n/w)/(1+\rho+\delta)\bigr)=l\}}\right|\\
&\le \sup_{t\in [0,-\log\eta/(1+\rho+\delta)]}\PP\left(\xi^{(2)}_{2,\delta}(t+\epsilon)-
 \xi^{(2)}_{2,\delta}\bigl((t-\epsilon)\wedge 0)\bigr)\ge 1\right)\\
&\quad +\PP\left(\sup_{n\eta \le w\le n}\left|T_n-T_{w-1}-\frac{1}{1+\rho+\delta}\log(n/w)\right|\ge \epsilon\right)=: p_2(\epsilon,\eta).
 \end{align*}
 Then by \eqref{eq:A1_bound}, we see that
 \begin{align*}
 \EE|A_1(n)|&\le 2\cdot\frac{1}{n}\cdot n\eta+\frac{1}{n}(1-\eta)n\bigl(p_1(\epsilon,\eta)+p_2(\epsilon,\eta)\bigr),
 \end{align*}
 which implies $\lim_{n\to\infty}\EE|A_1(n)|=0$. Therefore, $A_1(n)\convp 0$.
 
 For $A_2(n)$, we use Markov's inequality to obtain
 \begin{align}
& \PP(|A_2(n)|>\epsilon)\nonumber\\
 \le& \frac{1}{n^4\epsilon^4}
 \EE\left[\sum_{w=2}^{n+1} \left(\ind_{\left\{\bxi_{w,\delta}\left(\frac{\log(n/w)}{1+\rho+\delta}\right)=(m,l)\right\}}
 - \PP\left(\bxi_{w,\delta}\left(\frac{\log(n/w)}{1+\rho+\delta}\right)=(m,l)\right)\right)\right]^4.
 \label{eq:A2_bound}
 \end{align}
Write $X_w:= \ind_{\left\{\bxi_{w,\delta}\left(\frac{\log(n/w)}{1+\rho+\delta}\right)=(m,l)\right\}}
 - \PP\left(\bxi_{w,\delta}\left(\frac{\log(n/w)}{1+\rho+\delta}\right)=(m,l)\right)$, then $\EE(X_w)=0$.
 Also, by the independence among $\{\bxi_{w,\delta}(\cdot)\}_{w\ge 2}$, we have
 $$\EE(X_wX_u)=\EE(X_w^2X_u)=\EE(X_w^3X_u)=0,\qquad w\neq u. $$
 Then \eqref{eq:A2_bound} becomes
 \begin{align*}
& \PP(|A_2(n)|>\epsilon)\le\frac{1}{n^4\epsilon^4}\left[\sum_{w=2}^{n+1} X_w\right]^4\\
 =&\frac{1}{n^4\epsilon^4} \EE\left[\sum_{w} X_w^4+4\sum_{w_1\neq w_2}X_{w_1}^3 X_{w_2}
 + 3\sum_{w_1\neq w_2}X_{w_1}^2 X_{w_2}^2 
 \right.\\
 &\left.\qquad
 +6\sum_{w_1\neq w_2\neq w_3}X_{w_1}^2 X_{w_2} X_{w_3}
 + \sum_{w_1\neq w_2\neq w_3\neq w_4}X_{w_1} X_{w_2}X_{w_3}X_{w_4}\right]\\
 =& \frac{1}{n^4\epsilon^4} \EE\left[\sum_{w} X_w^4
 + 3\sum_{w_1\neq w_2}X_{w_1}^2 X_{w_2}^2 \right]
 \le \frac{1}{n	^3\epsilon^4}+\frac{1}{n^2\epsilon^4}.
 \end{align*}
 Then by the Borel-Cantelli lemma, we have $A_2(n)\convas 0$.
 For $A_3(n)$, since the function $\PP[\bxi_{2,\delta}(t)=(m,l)]$ is bounded and continuous in $t$, 
 then $A_3(n)\to 0$ by the Riemann integrability of $\PP[\bxi_{2,\delta}(-\log t/(1+\rho+\delta))=(m,l)]$,
 thus completing the proof of \eqref{eq:Nij}.

\end{proof}

\medskip

Based on the limiting joint distribution in Theorem~\ref{thm:limitNij}, we further study the asymptotic behavior of $(\mathcal{I},\mathcal{O})$,
and the next theorem shows that $(\mathcal{I},\mathcal{O})$ are jointly regularly varying. 

\begin{Theorem}\label{thm:IO_mrv}
Let $(\mathcal{I},\mathcal{O})$ be as in \eqref{eq:Nij}.
If $\lambda_1\ge \log2$, then
\begin{align}
\label{eq:MRV_IO}
(\mathcal{I},\mathcal{O}) \in \text{MRV}\left(\frac{1+\rho+\delta}{\lambda_1}, t^{\lambda_1/(1+\rho+\delta)}, \mu, \mathbb{R}_+^2\setminus\{\origin\}\right),
\end{align}
where the 
limit measure $\mu \in \mathbb{M} (\R_+^2\setminus \{\bzero\})$
satisfies for any $f\in \mathcal{C} (\R_+^2\setminus \{\bzero\})$,
\begin{align}\label{eq:limit_mu}
  \mu(f) =\int_0^\infty \EE \bigl( f(y\widetilde L
  \bv) \bigr)\nu_{(1+\rho+\delta)/\lambda_1} (\dd y),
\end{align}
and $\widetilde{L}$ satisfies
$e^{-\lambda_1t}\widetilde{\bxi}_\delta(t)\convas \widetilde{L}\bv$.
Also, $(v_1,v_2)^T\equiv \bv$ is given in \eqref{eq:defv}. Since
$\widetilde L$ is one-dimensional and $\bv$ is deterministic, the
distribution of $\widetilde L \bv$ concentrates on a one-dimensional
subspace and therefore
$\mu(\cdot)$ concentrates Pareto mass on the line $y=ax$ where
\begin{equation}\label{e:defa}
  a=\frac{v_2y\widetilde L}{v_1y\widetilde L} =\frac{v_2}{v_1} =
\begin{cases}
  \frac{  \gamma-\alpha+\sqrt \Delta}{2\gamma\rho},&\text{ if }\alpha<1,\\
  \rho,& \text{ if } \alpha=1.
\end{cases}
\end{equation}
In addition, there exists some constant $C>0$ such that
$\mu\bigl((x,\infty)\times [0,\infty)\bigr) = C
x^{-(1+\rho+\delta)/\lambda_1}$, $x>0$.

Switching to $L_1$-polar coordinates via the transformation
$$T:(x,y)\mapsto \Bigl(\frac{(x,y)}{x+y}, (x+y) \Bigr) =: (\btheta, r)$$
from $\R_+^2 \setminus \{\bzero\} \mapsto  \{(x,y) \in
\R_+^2\setminus \{\bzero\} : x+y=1\} \times
(0,\infty)=:\aleph_0\times (0,\infty)$, we find with
$$\btheta_0=\Bigl( \frac{v_1}{v_1+v_2},\frac{v_2}{v_1+v_2} \Bigr)$$
that
$$\mu\circ T^{-1} (\dd\btheta, \dd r)=
\epsilon_{\btheta_0} (d\btheta) \widetilde{C} \nu_{(1+\rho+\delta)/\lambda_1}
(\dd r)$$
where $\epsilon_{\btheta_0}(\cdot)$ is the Dirac probabilty measure
concentrating all mass on $\btheta_0$ and
\[
\widetilde{C}= \int_0^\infty \frac{(1+\rho+\delta)}{\lambda_1}z^{-1-(1+\rho+\delta)/\lambda_1}
\times\PP\left(\widetilde{L}>\frac{1}{z(v_1+v_2)}\right)\dd z.
\]
\end{Theorem}

Note the connection between Cartesian and polar representations of $
\mu(\cdot)$ is
$$\btheta_0=\Bigl(\frac{1}{1+a}, \frac{a}{1+a}\Bigr),$$
and when $\alpha=1$, $\btheta_0$ simplifies to
$$\btheta_0=\Bigl (\frac{1}{1+\rho}, \frac{\rho}{1+\rho}\Bigr).$$

\begin{proof}
We will prove Theorem~\ref{thm:IO_mrv} by
 applying Theorem~\ref{th:extendBrei}, so we first need to check $\PP(\widetilde{L}>0)=1$.
 By the representation in \eqref{eq:conv_bxi2}, it suffices to show that $W_0, \widetilde{W}_k, k\ge 1$ are all strictly positive a.s., which follows directly from
 Theorem~V.7.2 and Equation~(V.25) in \cite{athreya:ney:1972}. Hence, we have $\widetilde{L}>0$ a.s..
 
From \eqref{eq:Nij}, we see that
$
\big(\mathcal{I},\mathcal{O}\bigr) 
\stackrel{d}{=} \widetilde{\bxi}(T^*)$,
where $T^*$ is an exponential random variable with rate
$1+\rho+\delta$, independent from the $\widetilde{\bxi}_\delta(\cdot)$
process.  The proof of 
 \eqref{eq:MRV_IO} and \eqref{eq:limit_mu} is an application of Theorem
 \ref{th:extendBrei} after making the identifications
 \begin{align*}
&\bxi (t)=t^{-1} \widetilde{\bxi}_\delta\left(\frac{1}{\lambda_1} \log t\right), 
                                                                      &&\bxi_\infty =\widetilde{L}\bv,
   &&X=e^{\lambda_1 T^*},\\
   &b(t)=t^{\lambda_1/(1+\rho+\delta)},&&
c=(1+\rho+\delta)/\lambda_1 && {}.
 \end{align*}
 The remaining piece is to show \eqref{e:extraCond} in this
 context and we will show
any $\delta \ge 0$ and any $q=1,2,\ldots$, there exists some constant $C(\delta,q)>0$ such that 
\begin{align}\label{eq:claim_moment}
\sup_{t\ge 0}e^{-\lambda_1 qt} \EE\left[\left(\widetilde{\xi}^{(1)}_\delta(t)\right)^q\right]
\le C(\delta, q).
\end{align}
We defer the technical proof of \eqref{eq:claim_moment} to
Section~\ref{sec:append}. The comments about where $\mu(\cdot)$
concentrates and the representation of $\mu $ in polar coordinates is
standard; see, for example, \cite[p. 292]{lindskog:resnick:roy:2014}
and this 
completes the proof of Theorem~\ref{thm:IO_mrv}.
\end{proof}

When $\alpha=1$, recall the special structure in \eqref{eq:A_alpha_1}, and we have
\begin{align}
\label{eq:xi12_a1}
\widetilde{\xi}^{(2)}_{\delta}(T^*) = \sum_{k=1}^{\widetilde{\xi}^{(1)}_{\delta}(T^*)} R_k,
\end{align}
and $\{\widetilde{\xi}^{(1)}_\delta(t):t\ge 0\}$ is a single-type birth immigration process with 
$\widetilde{\xi}^{(1)}_\delta(0)=0$ and transition rate $q_{h,h+1}=h+\delta$, $h\ge 0$.
For fixed $t$, applying the distributional property of a single-type birth immigration process (cf. \cite{tavare:1987}) we have that 
\begin{align}
\label{eq:dist_xi1}
\widetilde{\xi}^{(1)}_\delta(t)\stackrel{d}{=} (1-R)Z_\delta(e^{-t}) + R\,Z_{1+\delta}(e^{-t}),
\end{align}
where $Z_\delta(p)$, $p\in [0,1]$, is a negative binomial random variable with generating function
\[
\EE(s^{Z_\delta(p)}) = (s+(1-s)/p)^{-\delta},\qquad s\in [0,1],
\]
and $R$ is a Bernoulli random variable with $\PP(R=1)=\rho$, independent from $Z_\delta(e^{-t})$,
$Z_{1+\delta}(e^{-t})$ and $T^*$.
The next corollary gives the explicit asymptotic limit of $N_{m,l}(n)/n$ for $\alpha=1$.

\begin{Corollary}\label{cor:Nij}
When $\alpha=1$,
there exist some random variables $(\mathcal{I},\mathcal{O})$ such that for $m\ge 0$ and $l\ge 1$,
\[
\frac{N_{m,l}(n)}{n}\stackrel{p}{\longrightarrow} \PP(\mathcal{I}=m, \mathcal{O}=l).
\]
In particular, the generating function of $(\mathcal{I},\mathcal{O})$ is
\begin{align}
\mathbb{E}\left(z_1^{\mathcal{I}} z_2^{\mathcal{O}}\right)
=& (1-\rho)z_2
\int_0^\infty(1+\rho+\delta)e^{-(1+\rho+\delta)t}\bigl(z_1(1-\rho+\rho u_2)\nonumber\\
&\qquad \qquad +(1-z_1(1-\rho+\rho z_2))e^t\bigr)^{-\delta}\mathrm{d}t\nonumber\\
& + \rho z_1z_2 
\int_0^\infty(1+\rho+\delta)e^{-(1+\rho+\delta)t}\bigl(z_1(1-\rho+\rho z_2)\nonumber\\
&\qquad \qquad+(1-u_1(1-\rho+\rho z_2))e^t\bigr)^{-(1+\delta)}\mathrm{d}t,\label{eq:pgfXY}
\end{align}
for $ z_1, z_2\in [0,1]$.
Moreover, 
\begin{align*}
(\mathcal{I},\mathcal{O}) \in \text{MRV}\left({1+\rho+\delta}, t^{1/(1+\rho+\delta)}, \mu_1, \mathbb{R}_+^2\setminus\{\origin\}\right),
\end{align*}
where the
limit measure $\mu_1$ satisfies for $(x,y)\in \mathbb{R}_+^2\setminus \{\origin\}$,
\begin{align}\label{eq:IO_a1}
&\mu_1\bigl((x,\infty)\times (y,\infty)\bigr)\nonumber\\
&=\int_0^\infty {(1+\rho+\delta)}z^{-(2+\rho+\delta)}
\times\PP\left({L_0}>\frac{1}{z}\max\left\{x,y/\rho\right\}\right)\dd z,
\end{align}
and the pdf of ${L_0}$ is as given in \eqref{eq:pdf_L}. 
\end{Corollary}

\begin{proof}
The generating function in \eqref{eq:pgfXY} follows from the distributional representation in 
\eqref{eq:xi12_a1} and \eqref{eq:dist_xi1}.
Also, note that
\[
e^{-t} \widetilde{\xi}^{(1)}_{\delta}(t) \convas R\Gamma_{1+\delta}+(1-R)\Gamma_\delta,
\]
where $\Gamma_{1+\delta}$ and $\Gamma_{\delta}$ are two independent Gamma random variables with pdfs
\[
\frac{x^\delta e^{-x}}{\Gamma(1+\delta)},\quad \text{and}\quad
\frac{x^{\delta-1} e^{-x}}{\Gamma(\delta)},\qquad x\ge 0,
\] respectively,
and are also independent from $R$.
Then the convergence in \eqref{eq:IO_a1} is a direct result of Theorem~\ref{thm:limitNij}.
 \end{proof}

\subsection{Asymptotic dependence}

\subsubsection{Comments and simulations on asymptotic dependence.}\label{sub:sim}
The asymptotic dependence
structure 
for the limiting random variables $(\cal I, \cal O)$ is given by the
measure $\mu(\cdot)$ in Theorem \ref{thm:IO_mrv}.
The fact that $\mu $ concentrates on the  line
$\mathcal{L}_a:=\left\{(x,y)\in (0,\infty)^2: y= ax\right\},$
shows that 
 large $(\cal I, \cal O)$ pairs are fully dependent, a situation
 described in \cite{das:resnick:2017} as {\it full asymptotic dependence\/}.

\begin{figure}[h]
\centering
\includegraphics[scale=.45]{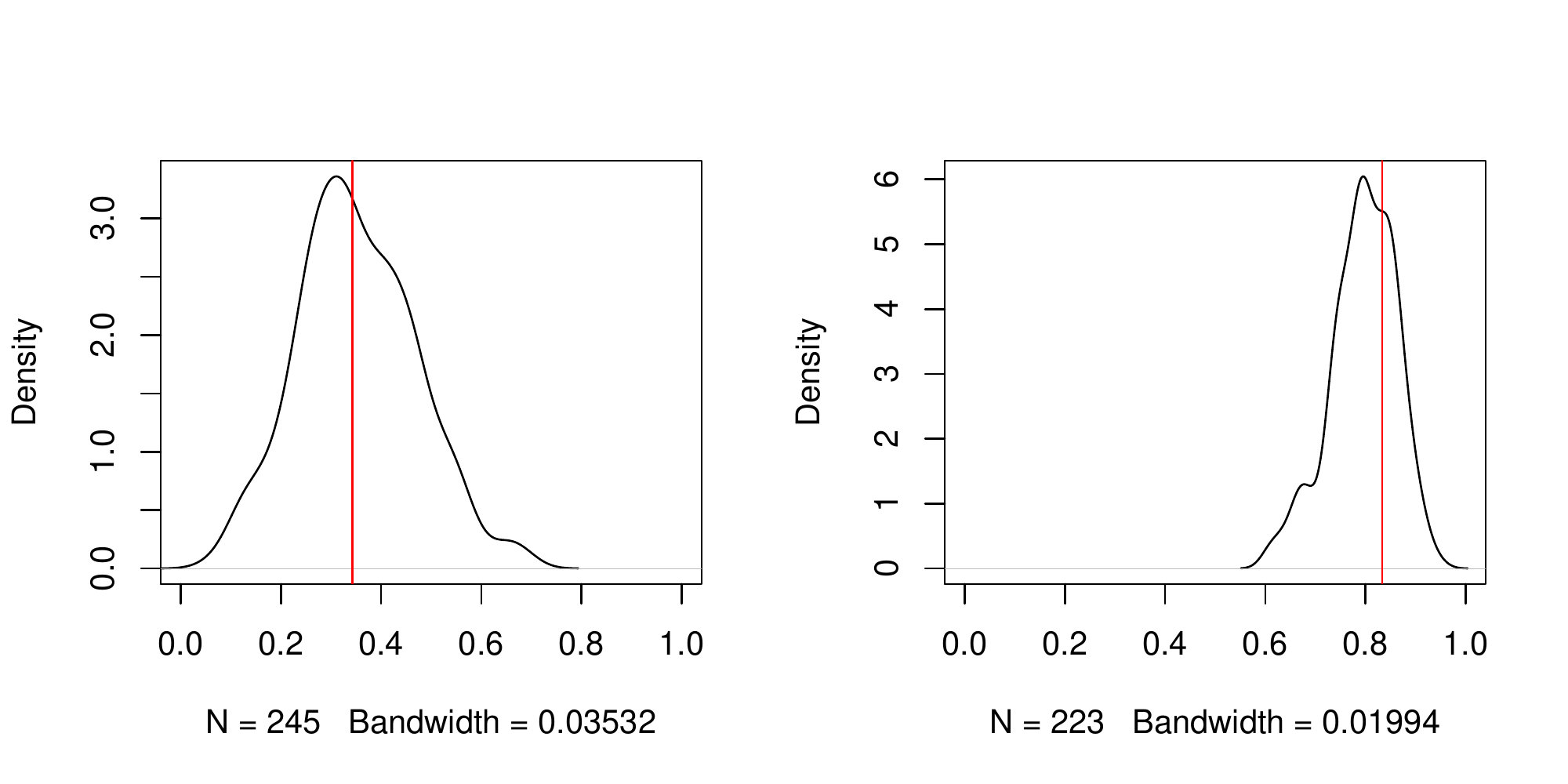}
\caption{Angular density plots for $(\alpha,\gamma,\rho, \delta,n) = (0.3,0.7,0.2, 2, 50000)$ (left) and $(1,0,2,0.2,50000)$ (right). The red vertical lines are placed at $1/(1+a)$, with $a$ defined in \eqref{e:defa}.}\label{fig:ang_b0}
\end{figure}

We illustrate  the full asymptotic dependence between in- and
out-degrees by simulation.
In Figure~\ref{fig:ang_b0}, we provide two empirical angular density plots for two sets of parameters: 
$(\alpha,\gamma,\rho, \delta,n) = (0.3,0.7,0.2, 2, 50000)$ (left) and $(1,0,0.2,2,50000)$ (right). 
The thresholds chosen in both cases are the 99.5\%-percentile of
$\{\Din_w(n)+\Dout_w(n): 1\le w\le n+1\}$. From
Figure~\ref{fig:ang_b0}, we see that both angular density plots show a
mode which is close to the red vertical lines placed at $1/(1+a)$, and
$a$ is as specified in \eqref{e:defa}. 
Comparing the two panels in Figure~\ref{fig:ang_b0}, we observe that
the $\alpha=1$ case shows
less spread about the vertical line at $a/(1+a).$

\subsubsection{Speculation on hidden regular variation.}\label{subsec:HRV}
When a limit measure
concentrates on a subcone of the full state space,
{to improve estimates of probabilities in the complement of the subcone,}
we can seek a
second {\it hidden\/} regular variation regime after removing the
subcone. In the present case, $\mu(\cdot)$ concentrates on
$\mathcal{L}_a$ and thus we may seek a regular variation property on
$\R_+^2\setminus \mathcal{L}_a$ using a weaker scaling function
$b_0(t)$ such that $t^{\lambda_1/(1+\rho+\delta)} /b_0(t) \to
\infty$. See \cite{das:resnick:2017, DasRes2015,
  resnickbook:2007,das:mitra:resnick:2013,lindskog:resnick:roy:2014}. A
convenient way to seek the hidden regular variation is by using {\it
  generalized polar coordinates\/} which in this case amount to the
  transformation
  $$\bx \to \Bigl(d_2(\bx, \EL _a), \frac{\bx}{d_2(\bx, \EL _a)}\Bigr),$$
where $d_2(\bx,\by) $ is a metric on $\R_+^2\setminus \{\bzero\}$
chosen here for convenience to be the $L_2$-metric.
The distance of a point $(x,y)$ to $\EL_a$ is readily computed
\cite[p. 881]{das:resnick:2017} to be
{$d_2 \bigl((x,y),\EL_a \bigr) ={|y-ax|}/{\sqrt{1+a^2}},$}
{and we use a scaled version }
\begin{equation}\label{e:d'}
d'_2 \bigl((x,y),\EL_a \bigr) =|y-ax|.\end{equation}
 Hidden regular
variation will be present for $(\mathcal{I},\mathcal{O}) \stackrel{d}{=} \widetilde
\bxi_\delta (T^*)$ if
$$tP\left[ \left(\frac{d'_2(\widetilde \bxi_\delta (T^*),\EL_a )}{b_0(t)},
\frac{\widetilde \bxi_\delta (T^*)}{d'_2\left(\widetilde \bxi_\delta
  (T^*),\EL_a \right)} \right) \in \cdot \,\right]
$$
converges to a limit measure in
$\mathbb{M}((0,\infty)\times \aleph_{\EL_a})$ where
$\aleph_{\EL_a} =\{\bx \in \R_+^2 \setminus \EL_a: d_2(\bx, \EL_a
)=1\}$. {Thus, considering \eqref{e:d'},} evidence consistent with this hidden regular variation is that 
$$d'_2(\widetilde \bxi_\delta (T^*),\EL_a ) =\left|\widetilde{\xi}_\delta^{(2)}(T^*) -a\widetilde{\xi}_\delta^{(1)}(T^*)\right|$$ be a random variable with a
regularly varying tail. Since $\widetilde{\bxi}_\delta (t)/e^{\lambda_1t} \convas \widetilde{L}\bv$, we have
$$ \frac{ \xi_\delta^{(2)}(t)  -a\xi_\delta^{(1)}(t) }{e^{\lambda_1t}} \convas  \widetilde{L}v_2-a \widetilde{L} v_1=0,$$
so the index of the sought hidden regular variation needs to be
smaller than the one indicated in \eqref{eq:MRV_IO}.
In what follows, we {offer} some {(incomplete) evidence this is the case.}

When $\delta=0$ and $\alpha<1$, $\widetilde{\bxi}_0(\cdot)$ becomes a two-type branching process without immigration.
Then by \eqref{eq:defA}, the second eigenvalue of $A$ is 
\[
\lambda_2=\frac{1}{2}(1-\sqrt{\Delta}),
\]
and its associated right eigenvector is $\bu'= (a,-1)^T$.
Note that Theorem~\ref{thm:limitNij} assumes $\lambda_1\ge \log 2$, i.e. $\sqrt{\Delta}\ge 2\log2-1$,
so that
$$\lambda_1-2\lambda_2=\frac{3}{2}\sqrt{\Delta}-\frac{1}{2}=3\log 2-2>0.$$
By \cite[Corollary V.8.1]{athreya:ney:1972}, one property for a 2-type branching process, $\widetilde{\bxi}_0(\cdot)$, is that if
\begin{align}
\lambda_1>2\lambda_2,
\label{eq:lambda12}
\end{align}
then 
there exists some normal random variable $V$ with mean 0 and
{some} variance $\sigma^2>0$, such that 
\begin{align}
\lim_{t\to\infty}& \PP\left(0<x_1<\widetilde{L}\le x_2<\infty,\, e^{-\lambda_1t/2}{\left|\widetilde{\xi}^{(2)}_0(t)-a\widetilde{\xi}^{(1)}_0(t)\right|}\le y\right)\nonumber\\
&= \int_{x_1}^{x_2} f_{\widetilde{L}}(z)
\PP\left(|V|\le \frac{y}{\sqrt{z}}\right)\dd z.
\label{eq:second_IO}
\end{align}
In \eqref{eq:second_IO}, the distribution of $\widetilde{L}$ {as well
as the constant $\sigma$} depend on ${\widetilde{\bxi}_0(0)}$, but for brevity of notation, we suppress the conditioning on ${\widetilde{\bxi}_0(0)}$.
Based on \eqref{eq:second_IO}, we make the following conjecture: if it were true that for $q'>2(1+\rho)/\lambda_1$,
\begin{align}
\label{eq:moment2}
\sup_{t\ge 0} e^{-\lambda_1 q' t/2} \EE\left(\left| \widetilde{\xi}^{(2)}_0(t)-a\widetilde{\xi}^{(1)}_0(t)\right|^{q'}\right) <\infty,
\end{align}
then 
\begin{align}
\label{eq:conjecture}
t\PP\Bigl(&\frac{\left|\mathcal{O}-a\mathcal{I}\right|}{t^{\lambda_1/(2(1+\rho))}}>y\Bigr)
=
  t\PP\left(\frac{\left|\widetilde{\xi}^{(2)}_0(T^*)-a\widetilde{\xi}^{(1)}_0(T^*)\right|}{t^{\lambda_1/(2(1+\rho))}}>y\right)\nonumber\\ 
& \longrightarrow \frac{1+\rho}{\lambda_1}\int_0^\infty u^{-1-(1+\rho)/\lambda_1}
\int_{0}^{\infty} f_{\widetilde{L}}(z)
\PP\left(|V|> \frac{y}{\sqrt{uz}}\right)\dd z,\qquad {(t\to\infty)}.
\end{align}

{We have not been able to prove the} moment condition
\eqref{eq:moment2}.
{However, we offer simulation evidence for \eqref{eq:conjecture}}
{and simulate} a network with parameters
\[
(\alpha,\gamma, \rho, \delta, n) = (0.2,0.8,0.25,{0},50000),
\]
so that $\lambda_1=0.816>\log 2$, and $\lambda_2 = 0.184<\lambda_1/2$.
\begin{figure}[h]
\centering
\includegraphics[scale=.4]{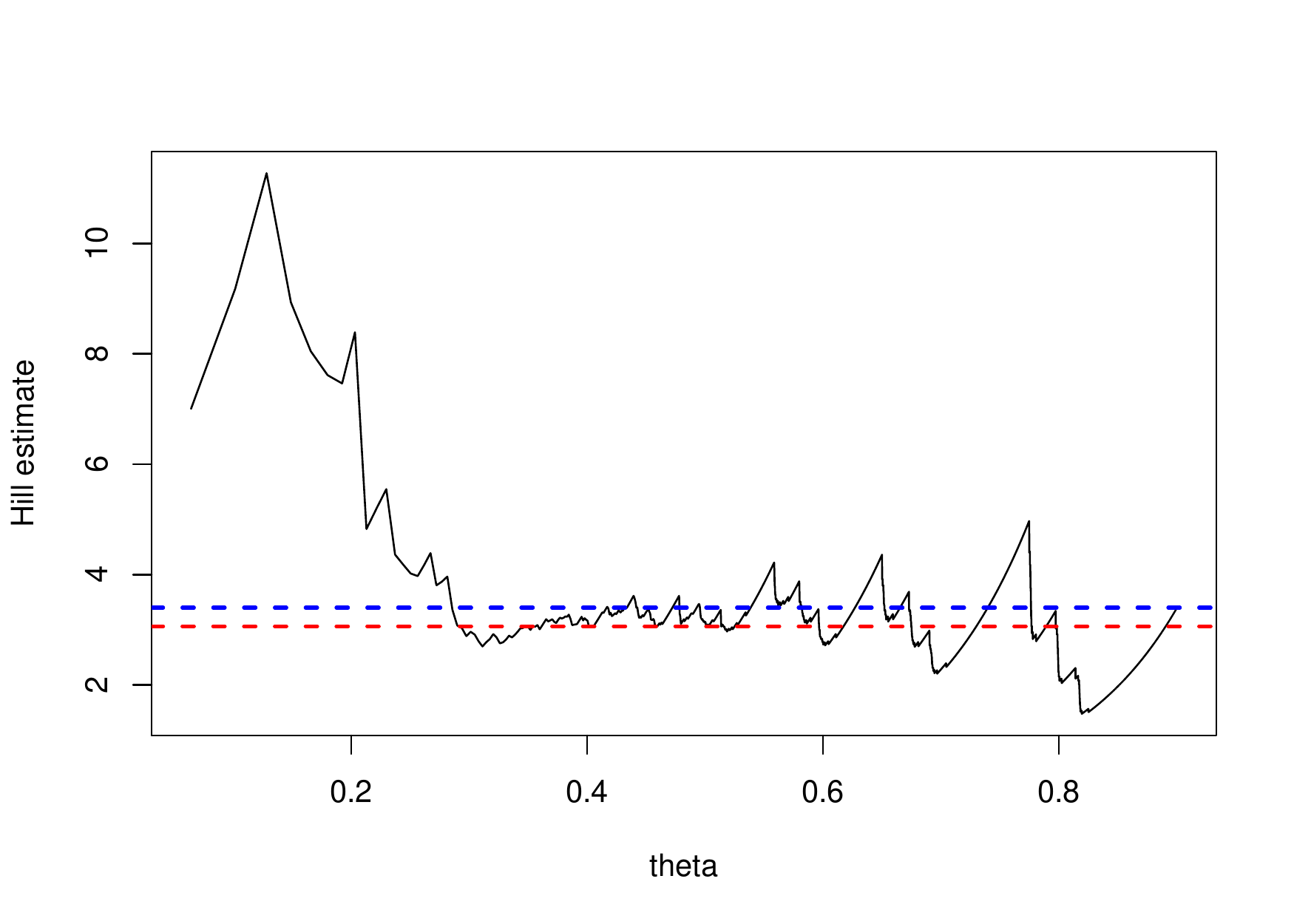}
\caption{
AltHill plot based on \eqref{pt:altHill} with $(\alpha,\gamma,\rho, \delta, n) =  (0.2,0.8,0.25,0,50000)$. The red and blue lines represent the true and estimated values of $2(1+\rho)/\lambda_1$, respectively.}\label{fig:altHill0}
\end{figure}
In Figure~\ref{fig:altHill0}, we present the altHill plot (cf. \cite{drees:dehaan:resnick:2000}) based on
\begin{align}\label{pt:altHill}
\left\{\left|a\Din_w(n)-\Dout_w(n)\right|: w\in V(n)\right\}.
\end{align}
The red dashed line is the true value of $2(1+\rho)/\lambda_1$, and
the blue line corresponds to the estimated value of $2(1+\rho)/\lambda_1$ using the minimum distance method in \cite{clauset:shalizi:newman:2009}. 
The altHill plot supports our conjecture in \eqref{eq:conjecture}.
For $\delta>0$, we have similar speculation as in \eqref{eq:conjecture}, i.e. we conjecture
$\left|\mathcal{O}-a\mathcal{I}\right|$ be regularly varying with index $2(1+\rho+\delta)/\lambda_1$.
We repeat the simulation procedure above for a different set of parameters:
\[
(\alpha,\gamma, \rho, \delta, n) = (0.2,0.8,0.25,0.5,50000).
\]
The corresponding altHill plot is given in Figure~\ref{fig:altHill1}.
The red dashed line is the true value of $2(1+\rho+\delta)/\lambda_1$, and
the blue line corresponds to the estimated value of $2(1+\rho+\delta)/\lambda_1$.
These provide
evidence for our conjecture, even when $\delta=0.5>0$.

\begin{figure}[h]
\centering
\includegraphics[scale=.4]{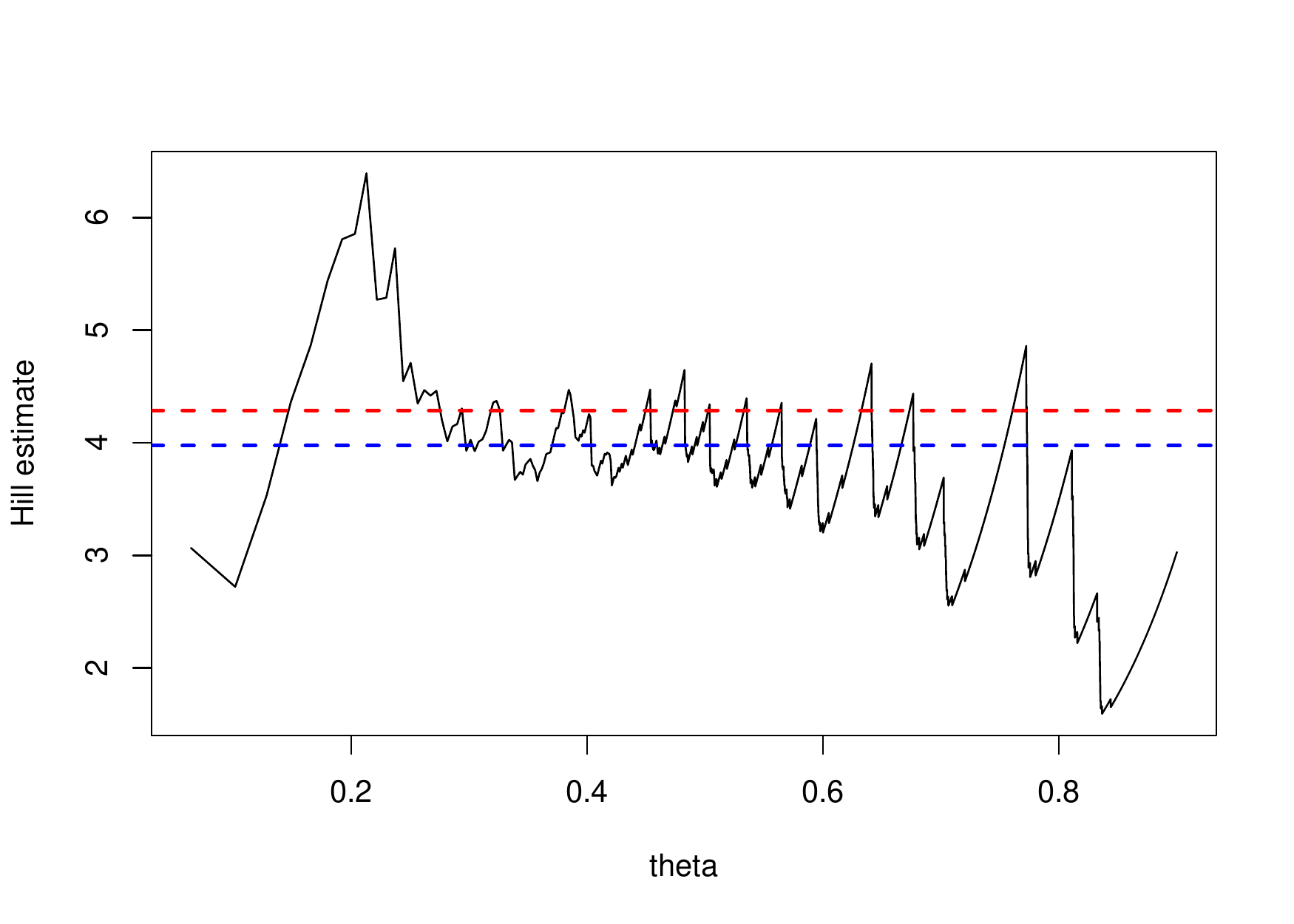}
\caption{
AltHill plot based on \eqref{pt:altHill} with $(\alpha,\gamma,\rho, \delta, n) =  (0.2,0.8,0.25,0.5,50000)$. The red and blue lines represent the true and estimated values of $2(1+\rho+\delta)/\lambda_1$, respectively.}\label{fig:altHill1}
\end{figure}

\section{Potential Extensions of the Model}\label{sec:discuss}
Now that we have studied theoretical properties of the proposed PA
model with reciprocity, this section discusses possible extensions
that we will investigate in the future.

\paragraph{Randomized $\rho$.}
So far we have assumed that each node has a common and fixed
reciprocity parameter, $\rho\in (0,1)$ but this may not reflect
realistic behavior and allow good fits to datasets.
For instance, users on Facebook (nodes) may
have their own probabilities to respond to a post on his/her wall. 
Keeping such heterogeneity in mind, we consider a randomized scenario: for each node $v\ge 1$, upon its creation, its affiliated reciprocity probability, $U_v$, is uniformly distributed on $[A,B]$. Here we also assume $0\le A\le B\le 1$, and $\{U_v:v\ge 1\}$ are iid.

\begin{figure}[h]
\centering
\includegraphics[scale=.4]{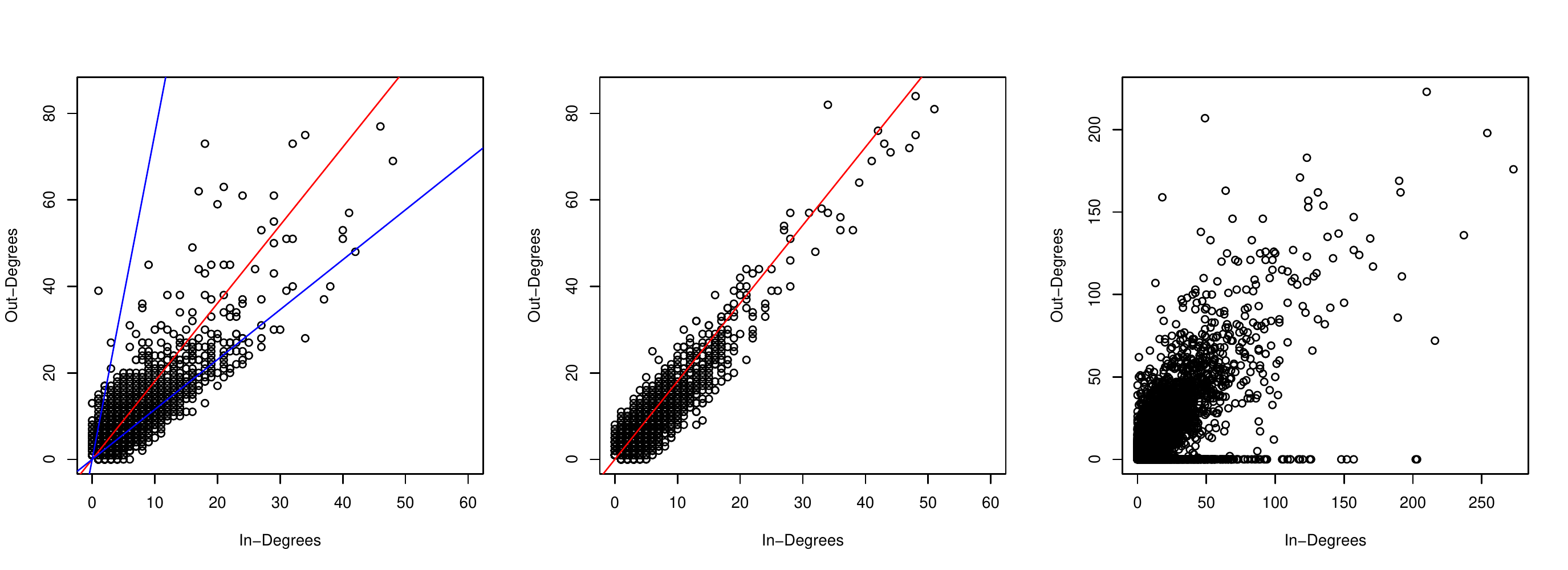}
\caption{Scatter plots of simulated in- and out-degrees from PA models with randomized (left) and fixed (middle) reciprocity, as well as those from Facebook wall posts between 2006-01-01 and 2007-02-25. We set $(\alpha,\gamma,\delta, n) =  (0.2,0.8, 0.5,10^5)$. The random reciprocity is uniform on $[0.1,0.8]$ and the fixed reciprocity $\rho=0.45$. Red lines in left and middle panels correspond to $\mathcal{L}_a= \{(x,y)\in\RR^2_+: y= 1.805 x\}$, and the two blue lines in the left panel constructs $[\text{wedge}]$ as in \eqref{eq:defw}.}\label{fig:random}
\end{figure}

We use simulation to inspect features of this PA model with randomized reciprocity. Set
$(\alpha,\gamma,\delta, n) =  (0.2,0.8, 0.5,10^5)$, and
 simulate a PA network with randomized reciprocity using 
$(A,B) =  (0.1, 0.8)$, and the scatter plot of simulated out- vs in-degrees is presented in the left panel of Figure~\ref{fig:random}.
Since in this case $\EE(U_1)=0.45$, for comparison purposes, we also include a scatter plot from a PA network with fixed reciprocity, i.e. $(\alpha,\gamma,\rho, \delta, n) =  (0.2,0.8,0.45, 0.5,10^5)$, in the middle panel of Figure~\ref{fig:random}.
When $\rho=0.45$, the slope of $\mathcal{L}_a$, is calculated from \eqref{e:defa}, i.e. $a=1.805$, and $\mathcal{L}_a$ is marked as the red line in left and middle panels of Figure~\ref{fig:random}. 
In the right panel of Figure~\ref{fig:random}, we also include the scatter plot of out- vs in-degrees from the Facebook wall posts data (available at \url{http://konect.cc/networks/facebook-wosn-wall/}) between 2006-01-01 and 2007-02-25.

The three panels in Figure~\ref{fig:random} show that a randomized reciprocity parameter tends to generate a scatter plot which matches the real data example better. 
In addition, based on the left and middle scatter plots, we see that in- and out-degrees from the model with randomized reciprocity  are less concentrating around $\mathcal{L}_a$.
The two blue lines in the left panel of Figure~\ref{fig:random} correspond to the 
concentration lines for $\rho=A,B$, whose slopes are equal to 7.533 and 2, respectively.
The left panel of Figure~\ref{fig:random} shows that large in- and out-degrees are concentrating within the wedge:
\begin{align}
\label{eq:defw}
[\text{wedge}] :=  \{(x,y)\in\RR^2_+: 1.154 x\le y\le 7.533 x\}.
\end{align}
This phenomenon coincides with the notion of \emph{strong asymptotic dependence} as discussed in \cite{das:resnick:2017}, and has been detected in real datasets(cf. \cite[Section~5]{das:resnick:2017}), but the rigorous justification under the network setup needs working out.

\paragraph{Adding edges between existing nodes.} 
The model proposed in Section~\ref{sec:model} requires adding a new node at each step, and the offset parameters are identical for the evolution of both in- and out-degrees. These assumptions are restrictive in practice.
We create an
 extended model by including the following scenario
added to cases (i) and (ii) in Section~\ref{sec:model}:
\begin{enumerate}
\item[(iii)] With probability $\beta\in (0,1)$, we add an edge $(w_1,w_2)$ between two existing nodes $w_1,w_2\in V(n)$, 
with probability
\begin{align*}
&\frac{\Din_{w_2}(n)+\delta}{\sum_{w_2\in V(n)} (\Din_{w_2}(n)+\delta)}\frac{\Dout_{w_1}(n)+\delta}{\sum_{w_1\in V(n)} (\Dout_{w_1}(n)+\delta)}\\
&= \frac{\Din_{w_2}(n)+\delta}{|E(n)|+\deltain |V(n)|}\frac{\Dout_{w_1}(n)+\delta}{|E(n)|+\delta|V(n)|}. 
\end{align*}
Then with probability $\rho\in (0,1)$, we add a reciprocal edge $(w_2,w_1)$. We then update 
the edge set as
$E(n+1) = E(n)\cup \{(w_1,w_2), (w_2,w_1)\}$. 
If the reciprocal edge is not created, then $E(n+1) = E(n)\cup \{(w_1,w_2)\}$. 
\end{enumerate}
With this third scenario, both the number of nodes and of edges in $G(n)$
are
random. 
The $\beta$-scenario  makes the model more realistic since a large proportion of edges are created between two existing nodes in real datasets (e.g. those listed on KONECT \cite{kunegis:2013}). However,
the additional scenario does not fit into the embedding framework directly, as it may simultaneously increase the in- and out-degrees of two different nodes.
In addition, the HRV conjecture in Section~\ref{subsec:HRV} is based on the condition that $\lambda_1>2\lambda_2$, which is a direct result from the $\lambda_1\ge \log 2$ assumption in Theorem~\ref{thm:limitNij}. With the additional $\beta$-scenario, it is possible to have 
$\lambda_1<2\lambda_2$, which may lead to different limiting behavior of $d'_2((\mathcal{I},\mathcal{O}),\EL_a )$.

\paragraph{Asymptotic dependence in higher dimensions.}
Another possible extension is to consider the multitype branching
structure with $K\ge 2$. For instance, when modeling the retweet or
reply network on Twitter, the follower/following relationship may also
have an impact on the generation of retweets or replies (directed
edges). Assume each node (user) to have four different degrees:
ordinary in- and out-degrees, together with the numbers of followers
and accounts one is following. Then an embedding framework using
four-type branching processes may provide a resolution to uncover the
asymptotic dependence structure between retweeting degrees and
friendship degrees.

We will explore further in all these directions in our future research. 

\section{Proofs}\label{sec:proof}

\subsection{Proof of Theorem~\ref{thm:mBI}}\label{sec:append_thm}
By \cite[Theorem~V.8.1]{athreya:ney:1972}, we see that
$$\left\{e^{-\lambda_t}\bu\cdot \bPsi_k(t): t\ge 0\right\}$$
is a martingale with respect to the filtration
$\sigma\left(\bPsi_k(s):0\le s\le t\right)$.
Also, discussion in Chapter~V.7.4 of \cite{athreya:ney:1972} shows that for $k\ge 0$,
\[
\sup_{t\ge 0}\EE\left[\left(e^{-\lambda_1 t}\bu\cdot \bPsi_k(t)\right)^2\right]<\infty.
\]
Therefore, $\{e^{-\lambda_1 t}\bu\cdot \bxi_\theta(t): t\ge 0\}$ is a non-negative sub-martingale 
with respect to $\sigma\left(\bxi_\theta(s):0\le s\le t\right)$, and
\begin{align*}
\sup_{t\ge 0} &\EE\left(e^{-\lambda_1 t} \bu\cdot\bxi_\theta(t)\right)
= \sup_{t\ge 0} \sum_{k=0}^\infty \EE\left(e^{-\lambda_1 \tau_k}
e^{-\lambda_1 (t-\tau_k)} \bu\cdot \bPsi_k(t-\tau_k)\ind_{\{t\ge \tau_k\}}
\right)\\
&\le \sum_{k=0}^\infty \left(\EE\left(e^{-2\lambda_1 \tau_k}\right)\right)^{1/2}
\sup_{t\ge 0}\left(\EE\left[\left(e^{-\lambda_1 t}\bu\cdot \bPsi_0(t)\right)^2\right]\right)^{1/2}
<\infty.
\end{align*}
Then we conclude that 
there exists some random variable $I_0$ such that as $t\to\infty$,
\begin{align*}
e^{-\lambda_1 t} \bu\cdot\bxi_\theta(t) \convas I_0.
\end{align*}
Also, since
\begin{align*}
&\left(e^{-\lambda_1 t} \bu\cdot\bxi_\theta(t)\right)^2\\
=& \left(\sum_{k=0}^\infty e^{-\lambda_1\tau_k}
e^{-\lambda_1 (t-\tau_k)} \bu\cdot \bPsi_k(t-\tau_k)\ind_{\{t\ge \tau_k\}}
 \right)^2\\
 \le& \left(\sum_{k=0}^\infty e^{-\lambda_1\tau_k}\right)
 \left(\sum_{k=0}^\infty e^{-\lambda_1\tau_k}\left(e^{-\lambda_1 (t-\tau_k)} \bu\cdot \bPsi_k(t-\tau_k)\ind_{\{t\ge \tau_k\}}\right)^2 \right),
\end{align*}
then 
\begin{align*}
&\EE\left[\left(e^{-\lambda_1 t} \bu\cdot\bxi_\theta(t)\right)^2\right]\\
\le& \sup_{t\ge 0}\EE\left[\left(e^{-\lambda_1 t}\bu\cdot \bPsi_0(t)\right)^2\right]
\EE\left[\sum_{k=0}^\infty e^{-\lambda_1\tau_k}
 \left(\sum_{k=0}^\infty e^{-\lambda_1\tau_k}\right)\right]\\
 \le& \sup_{t\ge 0}\EE\left[\left(e^{-\lambda_1 t}\bu\cdot \bPsi_0(t)\right)^2\right]
 \sum_{k=0}^\infty \left(\EE\left(e^{-2\lambda_1 \tau_k}\right)\right)^{1/2}
 \left(\EE\left[\left(\sum_{k=0}^\infty e^{-\lambda_1\tau_k}\right)^2\right]\right)^{1/2}\\
 <&\infty.
\end{align*}
Hence, we have
\begin{align*}
\EE\left(I_0\right)&= \lim_{t\to\infty} \EE\left(e^{-\lambda_1 t}\bu\cdot \bxi_\theta(t)\right)\\
&= \lim_{t\to\infty} \sum_{k=0}^\infty \EE\left(e^{-\lambda_1 \tau_k}\right)
\EE\left(e^{-\lambda_1 (t-\tau_k)} \bu\cdot \bPsi_k(t-\tau_k)\ind_{\{t\ge \tau_k\}}\right)\\
&= \sum_{k=0}^\infty \EE\left(e^{-\lambda_1 \tau_k}\right) \EE\left(W_k\right)
=\EE\left(\sum_{k=0}^\infty e^{-\lambda_1 \tau_k}W_k\right).
\end{align*}
In addition, by Fatou's lemma, we see that
a.s.
\[
I_0\ge \sum_{k=0}^\infty e^{-\lambda_1\tau_k} W_k.
\]
Then we conclude that
\[
I_0 = \sum_{k=0}^\infty e^{-\lambda_1\tau_k} W_k\qquad \text{a.s.}.
\]

Then for $i=1,\ldots, K$, we see that
\begin{align*}
e^{-\lambda_1t} \xi_\theta^{(i)}(t) 
&= e^{-\lambda_1t}\sum_{k=0}^\infty \frac{\psi_k^{(i)}(t-\tau_k) }{\bu\cdot \bPsi_k(t-\tau_k)}
\bu\cdot \bPsi_k(t-\tau_k)\ind_{\{t\ge \tau_k\}}.
\end{align*}
For fixed $M>0$, it follows from \eqref{eq:conv_bxi} that
\begin{align*}
e^{-\lambda_1t}\sum_{k=0}^M \frac{\psi_k^{(i)}(t-\tau_k) }{\bu\cdot \bPsi_k(t-\tau_k)}
\bu\cdot \bPsi_k(t-\tau_k)\ind_{\{t\ge \tau_k\}}
\convas \sum_{k=0}^M e^{-\lambda_1\tau_k}W_k v_i.
\end{align*}
Also, since all entries of $\bu$ are strictly positive, we then have
\begin{align*}
\lim_{M\to\infty}\lim_{t\to\infty} & e^{-\lambda_1 t}\sum_{k=M+1}^\infty \left|\frac{\psi_k^{(i)}(t-\tau_k) }{\bu\cdot \bPsi_k(t-\tau_k)}-v_i\right|
\bu\cdot \bPsi_k(t-\tau_k)\ind_{\{t\ge \tau_k\}}\\
&\le   \left(\frac{1}{u_i}+v_i\right) \lim_{M\to\infty}\lim_{t\to\infty}\sum_{k=M+1}^\infty
 e^{-\lambda_1 t}\bu\cdot \bPsi_k(t-\tau_k)\ind_{\{t\ge \tau_k\}} \stackrel{a.s.}{=} 0,
\end{align*}
which completes the proof of the theorem.

\subsection{Proof of Theorem~\ref{th:extendBrei}}\label{subsec:pf_breim}
To show (i) 
in $\mathbb{M}\left(\RR_+^p\times (\R_+\setminus \{0\})\right)$,
take
$f\in \mathcal{C}\bigl(\mathbb{R}_+^2\times (\R_+\setminus \{0\} )\bigr)$, which is bounded, uniformly continuous, and satisfies that for some $\varepsilon_0>0$,
$f(\bx, y) =0$ if $y<\varepsilon_0$. Then to prove
\eqref{e:beforeMult},
we  show 
\begin{align}
\label{eq:limit_IO}
\lim_{t\to\infty}&t\EE\left[f\left({{\bxi}(X)},
\frac{X}{b(t)}
                   \right)\right] =
                   \int_{\varepsilon_0}^\infty
                   \EE\left[f\bigl(\bxi_\infty, y\bigr)\right]
                   \nu_{c}(\dd y).
\end{align}
Rewrite the left hand side of \eqref{eq:limit_IO} as
\begin{align*}
& t \EE\left[f\left({\bxi(X)},
\frac{X}{b(t)}
\right)\right]
-t \EE\left[f\left(\bxi_\infty,
\frac{X}{b(t)}
\right)\right]\\
& \quad + t \EE\left[f\left(\bxi_\infty,
\frac{X}{b(t)}
\right)\right] - \int_{\varepsilon_0}^\infty
                   \EE\left[f\bigl(\bxi_\infty, y\bigr)\right]
                   \nu_{c}(\dd y)\\
& \quad + \int_{\varepsilon_0}^\infty \EE\left[f\bigl(\bxi_\infty,
                                      y\bigr)\right]
                                      \nu_{c}(\dd y)\\
&=: I_t+II_t + \int_{\varepsilon_0}^\infty
                                                         \EE\left[f\bigl(\bxi_\infty, y\bigr)\right]
                                                         \nu_{c}(\dd y),
\end{align*}
and we will show that both $|I_t|$ and $|II_t|$ go to 0 as $t\to\infty$.

For $II_t$, 
we see that
\begin{align*}
t &\EE\left[f\left(\bxi_\infty,
\frac{X}{b(t)}
\right)\right]
= \int_{\varepsilon_0}^\infty f\left(\bx,y\right) \PP\left(\bxi_\infty\in \dd \bx\right)
\times t\PP\left(\frac{X}{b(t)}\in \dd y\right).
\end{align*}
For any fixed $\bx\in \mathbb{R}_+^2$, $f(\bx, \cdot)\in
\mathcal{C}(\R_+\setminus \{0\})$ so that as $t\to\infty$,
\[
\int_{\varepsilon_0}^\infty f(\bx,y)\times t\PP\left(\frac{X}{b(t)}\in \dd y\right)\longrightarrow 
\int_{\varepsilon_0}^\infty f(\bx,y) \nu_{c}(\dd y).
\]
Since $f$ is bounded,
there exists some constant $K_0>0$ such that
$$\sup_{\bx\in \mathbb{R}_+^2, y\ge\varepsilon_0}f(\bx, y)\le K_0,$$ and
for all $t$,
\begin{align*}
\sup_{\bx\in \mathbb{R}_+^2, y\ge\varepsilon_0}&
\int_{\varepsilon_0}^\infty f(\bx, y)\times
                                                 t\PP\left(\frac{X}{b(t)}\in \dd y\right)\\ 
&\le K_0 \times t\PP\left(\frac{X}{b(t)}>\varepsilon_0\right)
\longrightarrow K_0 \varepsilon_0^{-c}<\infty.
\end{align*}
Then by dominated convergence, we see that $|II_t|\to 0$, as $t\to\infty$.

Now we consider the first term $I_t$. Since
$\bxi (t)  \to \bxi_\infty$ almost surely,  by Egorov's
Theorem, for any
$\varepsilon>0$, there exists an $\omega$-set $A_\varepsilon$ such
that 
\[
\sup_{\omega\in A_\varepsilon} \left|{\bxi}(t,\omega)-\bxi_\infty(\omega)\right|\longrightarrow 0,
\]
and $\PP(A_\varepsilon^c)<\varepsilon$.
Then we bound $|I_t|$ by
\begin{align}
|I_t|&\le \EE\left[\left|f\left({\bxi(X)}, \frac{X}{b(t)}\right)-
f\left(\bxi_\infty, \frac{X}{b(t)}\right)\right|\ind_{A_\varepsilon\cap\{{X}/{b(t)}\ge \varepsilon_0\}}\right]\nonumber\\
&\quad +\EE\left[\left|f\left({\bxi(X)}, \frac{X}{b(t)}\right)-
                                                                                                                            f\left(\bxi_\infty,\frac{X}{b(t)}\right)\right| \ind_{A^c_\varepsilon
 \cap X/b(t)\ge        \varepsilon_0\}}\right]\nonumber\\
&=: I^a_t+I^b_t.
\label{eq:bound_It}
\end{align}
Let $d\left((\bx, y), (\bx',y')\right)$ denote the distance between $(\bx, y)$ and $(\bx',y')$ on 
$\mathbb{R}_+^p\times \R_+\setminus \{0\}$ and recall the modulus of
continuity of $f$ given in 
\eqref{e:defModCon}.
Applying Egorov's Theorem, for a given $\eta>0$, there exists $M_0>0$ such that
\[
\sup_{t\ge M_0} \sup_{\omega\in A_\varepsilon}
\left|{\bxi}(t,\omega)-
  \bxi_\infty (\omega)\right|\le \eta.
\]
Then as long as $t$ is large enough such that $b(t)\varepsilon_0> M_0$, we have
\begin{align*}
  I^a_t &\le \Delta_f\left(\eta\right) t\PP\left(A_\varepsilon\cap
          \left\{\frac{X}{b(t)}\ge
          \varepsilon_0\right\}\right)\\ 
&\le \Delta_f\left(\eta\right) t\PP\left(\frac{X}{b(t)}\ge \varepsilon_0\right) \stackrel{t\to\infty}{\longrightarrow} \Delta_f\left(\eta\right) \varepsilon_0^{-c}
\stackrel{\eta\to 0}{\longrightarrow} 0.
\end{align*}
For $I^b_t$, we have
\begin{align*}
I^b_t &\le 2K_0\times t\PP\left(A_\varepsilon^c\cap\left\{\frac{X}{b(t)}\ge \varepsilon_0\right\}\right),\\
\intertext{and since $A^c_\varepsilon\in
  \sigma\{{\bxi}(t):t\ge 0\}$, and $X$ is
  independent of $\{{\bxi} (t):t\ge 0\}$, then} 
&= 2K_0 \PP\left(A_\varepsilon^c\right) t\PP\left(\frac{X}{b(t)}\ge
                                                   \varepsilon_0\right)
  \\
&\stackrel{t\to\infty}{\longrightarrow} 2K_0 \PP\left(A_\varepsilon^c\right)\varepsilon_0^{-c}
\le 2K_0 \varepsilon\cdot\varepsilon_0^{-c}.
\end{align*}
So  $|I_t|\to 0$, as $t\to\infty$, which completes the proof of 
\eqref{e:beforeMult}.

Now we turn to the proof of \eqref{e:prodOK} in part (ii).
Define a mapping $\chi: \RR_+^2\times (\R_+\setminus \{0\}) \mapsto
\RR_+^2\times (\R_+\setminus \{0\})$ by  
$
\chi((x,y), \sigma) = \bigl((\sigma x,\sigma y), \sigma\bigr).
$
Then by \cite[Example 3.3]{lindskog:resnick:roy:2014}, we apply $\chi$
to the convergence in \eqref{e:beforeMult} to obtain that in
$\mathbb{M}\left(\RR_+^2\times (\R_+\setminus \{0\})\right)$, 
\begin{align}
\label{eq:limit_IOchi}
 t\PP\left[\left( \frac{\bxi(X)X}{b(t)}, \frac{X}{b(t)} \right)\in\cdot\right]
\to \left(\PP\left[\bxi_\infty \in\cdot\right]\times \nu_{c}\right)\circ \chi^{-1}(\cdot).
\end{align}
Let $g\in \mathcal{C}(\RR_+^2\setminus\{\origin\})$, and
$g(\bx)=0$ if $\|\bx\| \leq \theta$ for some $\theta>0$ with
$\sup_{\bx }|g(\bx)|=:\|g\|.$ For typographical simplicity, assume $\|g\|=1.$
We can prove  \eqref{e:prodOK} by showing 
\begin{equation}\label{e:hopeTrue}
  t\EE g\Bigl( \frac{\bxi(X) X}{b(t)}\Bigr) \to
\int_0^\infty \EE
g(\bxi_\infty x) \nu_c (dx),
\end{equation}
knowing that from \eqref{eq:limit_IOchi}, for any $\eta>0$,
\begin{align}\label{e:witheta}
  t\EE& \left[g\Bigl(
\frac{\bxi(X) X}{b(t)}\Bigr)1_{[X/b(t)>\eta]}\right]\nonumber\\ 
&\to
\int_\eta^\infty \EE g\bigl(h(\bxi_\infty,x)\bigr) \nu_c(dx)
=\int_\eta^\infty \EE g(\bxi_\infty x) \nu_c(dx).
\end{align}
Differencing the left sides of \eqref{e:hopeTrue} and
\eqref{e:witheta} we must show
\begin{equation}\label{e:double}
  \lim_{\eta\to 0} \limsup_{t\to\infty}
  t\EE \left[g\Bigl(
  \frac{\bxi(X) X}{b(t)}\Bigr)\ind_{\{X/b(t)\leq\eta\}}\right]
  =0.\end{equation}
Bearing in mind the support of $g$, we have
\begin{align*}
  t\EE\left[ g\Bigl(
  \frac{\bxi(X) X}{b(t)}\Bigr)\ind_{\{X/b(t)\leq\eta\}}\right]
  \leq & t \PP \Bigl(
        \left\| \frac{\bxi(X) X}{b(t)}\right\| >\theta, \frac{X}{b(t)}
         \leq \eta \Bigr)\\
  &= t \PP \Bigl(
    \left\|      \frac{\bxi(X) X}{b(t)}\ind_{\{X/b(t)\leq \eta\}}
         \right \|>\theta \Bigr)\\
  \leq &  \theta^{-c'} t\EE \left \|   \frac{\bxi(X) X}{b(t)}\ind_{\{X/b(t)\leq \eta\}}
         \right \| ^{c'} \\
  \intertext{and recalling $\bxi(\cdot)$ and $X$ are independent, this
  is evaluated to}
  =&\theta^{-c'}\int_0^\eta \EE \left \|
{\bxi(b(t)s)}\right\|^{c'} s^{c'} tP(X/b(t) \in
     ds)\\
  \leq & \theta^{-c'}\kappa t\EE\left[\left( \frac{X}{b(t)}\right)^{c'} \ind_{\{X/b(t) \leq
  \eta\}}\right]\\
\intertext{and applying Karamata's theorem on integration, this converges as
  $t\to\infty$ }
  = & \theta^{-c'}\kappa\eta^{c'-c} \to 0 ,
    \end{align*}
as $\eta \to 0$ giving \eqref{e:double} and hence \eqref{e:hopeTrue}.

\subsection{Proof of Equation~\eqref{eq:claim_moment}}\label{sec:append}
We will now prove the claim in \eqref{eq:claim_moment} by showing a bound of the moment of a $K$-type branching process with immigration. We start with analyzing the $q$-th moment of a $K$-type branching process as given in Section~\ref{subsec:mBI}, for $q\ge 2$.
Then with the construction of a $K$-type branching process with immigration as in \eqref{eq:mBI_dist}, we derive the bound of its $q$-th moment.

\begin{Proposition}\label{prop:moment}
Let $\{\bPsi_0(t): t\ge 0\}$ be a continuous time $K$-type branching process with matrix $A$ as given in \eqref{eq:matrixA}, and $\bu$ be the right eigenvector associated with the largest eigenvalue of $A$, i.e. $\lambda_1$.
Suppose that $Y_j^{(i)}$ is the number of type $j$ particles produced by a type $i$ particle at the end of its lifetime, and 
\begin{align}\label{eq:cond_Yij}
Y_j^{(i)}\le e^{\lambda_1}, \qquad \text{for all }\quad i,j =1,\ldots,K.
\end{align}
Then for an integer $q\ge 1$, we have
\begin{align}
\label{eq:moment}
\sup_{t\ge 0}e^{-\lambda_1q t}\EE\left[\left(\bu \cdot \bPsi_0(t)\right)^q\right]<\infty.
\end{align}
\end{Proposition}

\begin{proof}

Note that for $n\ge 0$ integer, 
$\{\bPsi_0(n): n\ge 0\}$ is a discrete time $K$-type branching process, and
we have
\begin{align*}
\sup_{t\ge 0}e^{-\lambda_1q t}\EE\left[\left(\bu \cdot \bPsi_0(t)\right)^q\right]
\le e^{\lambda_1}\sup_{n\ge 0}e^{-\lambda_1q n}\EE\left[\left(\bu \cdot \bPsi_0(n)\right)^q\right].
\end{align*}
Hence it suffices to show for some constant $C'_q>0$,
\begin{align}
\label{eq:moment1}
\sup_{n\ge 0}e^{-\lambda_1q n}\EE\left[\left(\bu \cdot \bPsi_0(n)\right)^q\right]\le C'_q.
\end{align}
The proof of \eqref{eq:moment1} for $q=1,2$ has been given in \cite[Chapter V.7.4]{athreya:ney:1972}.
Here we use induction to show \eqref{eq:moment1} also holds for $q\ge 3$. 
Suppose \eqref{eq:moment1} is true for $k=1,\ldots,q$, $q\ge 2$, and define
\[
Z_n := e^{-\lambda_1 n}\left(\bu \cdot \bPsi_0(n)\right).
\]
We need to prove that for some $C'_{q+1}>0$,
\begin{align}
\label{eq:induc}
\sup_{n\ge 0}\EE\left[\left(Z_n\right)^{q+1}\right]\le C'_{q+1}.
\end{align}

Note also that by \cite[Theorem~V.8.1]{athreya:ney:1972}, $Z_n$ is a martingale with respect to the filtration $\mathcal{F}'_n:= \sigma\{Z_k:k\le n\}$ for $n\ge 0$.
Also, under $\PP^{\mathcal{F}'_n}$, 
\begin{align}
\label{eq:Zn}
Z_{n+1} \stackrel{d}{=}e^{-\lambda_1(n+1)} \sum_{j=1}^K u_j \sum_{i=1}^K \sum_{l=1}^{\psi_0^{(i)}(n)}Y^{(i)}_{j,l},
\end{align}
then we see 
from \eqref{eq:Zn} that for $q\ge 2$,
\begin{align}
\label{eq:bound_Zn}
\EE^{\mathcal{F}'_n}\left(Z_{n+1}^q\right)
&\le e^{-\lambda_1(n+1)q} \left(e^{\lambda_1}\bu\cdot \bPsi_0(n)\right)^q
= Z_n^q.
\end{align}
Next, we apply Lemma~2.1 in \cite{athreya:ghosh:sethuraman:2008} to obtain that
\begin{align*}
&\EE^{\mathcal{F}'_n}\left(Z_{n+1}^{q+1}\right)
\le \EE^{\mathcal{F}'_n}\left(Z_{n+1}^{q}\right) \EE^{\mathcal{F}'_n}\left(Z_{n+1}\right)\\
&\,+\sum_{k=1}^q{q\choose k}\EE^{\mathcal{F}'_n}\left(Z_{n+1}^{q-k}\right) 
e^{-\lambda_1(k+1)(n+1)} \sum_{j=1}^K u_j \sum_{i=1}^K \sum_{l=1}^{\psi_0^{(i)}(n)} \EE^{\mathcal{F}'_n}\left[\left(Y^{(i)}_{j,l}\right)^{k+1}\right],\\
\intertext{which by \eqref{eq:bound_Zn} is bounded by}
&\le  Z_n^{q+1}+\sum_{k=1}^q{q\choose k} Z_n^{q-k}e^{-\lambda_1(k+1)(n+1)} 
e^{\lambda_1(k+1)}\left(\bu\cdot \bPsi_0(n)\right)\\
&= Z_n^{q+1}+\sum_{k=1}^q{q\choose k} Z_n^{q-k+1}e^{-\lambda_1(k+1)(n+1)+\lambda_1 n} 
e^{\lambda_1(k+1)}\\
&=  Z_n^{q+1}+\sum_{k=1}^q{q\choose k}
Z_n^{q-k+1} e^{-\lambda_1 k n}.
\end{align*}
Then we have 
\begin{align}
\label{eq:Z_bound}
\sup_{n\ge 0}\EE\left(Z_{n}^{q+1}\right)&\le \EE\left(Z_{0}^{q+1}\right)
+ 
\sum_{k=1}^q{q\choose k} \sum_{n=0}^\infty e^{-\lambda_1kn} \EE\left(Z_n^{q-k+1}\right).
\end{align}
By the induction hypothesis, we have that $\sup_{n\ge 0}\EE\left(Z_n^{q-k+1}\right)\le C'_{q-k+1}$, $k=1,\ldots, q$. In addition,
$\sum_{n\ge 0}e^{-\lambda_1kn} <\infty$,
then \eqref{eq:induc} follows directly from \eqref{eq:Z_bound}.
\end{proof}

Using the distribution representation in \eqref{eq:mBI_dist}, we now give the moment bound for a 
$K$-type branching process with immigration. 
\begin{Proposition}\label{prop:moment_mBI}
Denote the $K$-type branching process with immigration given in 
\eqref{eq:mBI_dist} as
$$\left\{\bxi_\theta(t)\equiv (\xi^{(1)}_\theta(t),\ldots, \xi^{(K)}_\theta(t)): t\ge 0\right\}.$$
Suppose that the condition in \eqref{eq:cond_Yij} holds, then we have
\begin{align}
\label{eq:moment_imm}
\sup_{t\ge 0}e^{-\lambda_1q t}\EE\left[\left( \xi^{(i)}_\theta(t)\right)^q\right]<\infty.
\end{align}
\end{Proposition}
\begin{proof}
Let $\{\tau_l: l\ge 0\}$ be the jump times of a Poisson process with rate $\theta>0$, and $\tau_0\equiv 0$.
Assume that 
$$\left\{\bPsi_{0,l}(t)\equiv (\psi^{(1)}_{0,l}(t),\ldots, \psi^{(K)}_{0,l}(t)): t\ge 0\right\}_{l\ge 0}$$
 is a sequence of iid $K$-type branching process as in Proposition~\ref{prop:moment}, which is also independent from $\{\tau_l: l\ge 0\}$. Then by \eqref{eq:mBI_dist}, we construct a $K$-type branching process with immigration as 
\[
\bxi_\theta(t) =\sum_{l=0}^\infty\bPsi_{0,l}(t-\tau_l)\ind_{\{t\ge \tau_l\}}, \qquad t\ge 0.
\]
By Jensen's inequality, we see that
\begin{align}
\label{eq:jensen}
&\left(\frac{\sum_{l=0}^\infty e^{-\lambda_1(t-\tau_l)}\psi_{0,l}^{(i)}(t-\tau_l)\ind_{\{t\ge \tau_l\}}e^{-\lambda_1\tau_l}}{\sum_{l=0}^\infty e^{-\lambda_1 \tau_l}}\right)^q\nonumber\\
&\le \frac{\sum_{l=0}^\infty e^{-\lambda_1\tau_l}\left(e^{-\lambda_1(t-\tau_l)}\psi_{0,l}^{(i)}(t-\tau_l)\ind_{\{t\ge \tau_l\}}\right)^q}{\sum_{l=0}^\infty e^{-\lambda_1 \tau_l}}.
\end{align}
Write $\widetilde{Z}^{(i)}_l(t) := e^{-\lambda_1(t-\tau_l)}\psi_{0,l}^{(i)}(t-\tau_l)\ind_{\{t\ge \tau_l\}}$,
then \eqref{eq:jensen} implies
\begin{align*}
\left(\sum_{l=0}^\infty \widetilde{Z}^{(i)}_l(t) e^{-\lambda_1\tau_l}\right)^q
\le \left(\sum_{l=0}^\infty e^{-\lambda_1 \tau_l}\bigl(\widetilde{Z}^{(i)}_l(t)\bigr)^q\right)
\left(\sum_{k=0}^\infty e^{-\lambda_1\tau_k}\right)^{q-1}.
\end{align*}
Therefore,
\begin{align*}
e^{-\lambda_1q t}&\EE\left[\left( \psi^{(i)}_\theta(t)\right)^q\right]
= \EE\left[\left(\sum_{l=0}^\infty \widetilde{Z}^{(i)}_l(t) e^{-\lambda_1\tau_l}\right)^q\right]\\
&\le \EE\left[\left(\sum_{l=0}^\infty e^{-\lambda_1 \tau_l}\bigl(\widetilde{Z}^{(i)}_l(t)\bigr)^q\right)
\left(\sum_{k=0}^\infty e^{-\lambda_1\tau_k}\right)^{q-1}\right]\\
&\le \sup_{t\ge 0}\EE\left[\bigl(\psi_{0,1}^{(i)}(t)\bigr)^q\right]
\EE\left[\sum_{l=0}^\infty e^{-\lambda_1 \tau_l}\left(\sum_{k=0}^\infty e^{-\lambda_1\tau_k}\right)^{q-1}\right].
\end{align*}
By Proposition~\ref{prop:moment}, we have $\sup_{t\ge 0}\EE\left[\bigl(\psi_{0,1}^{(i)}(t)\bigr)^q\right]<\infty$, for all $i=1,\ldots,K$, then it suffices to show
\begin{align}
\label{eq:final_claim}
\sum_{l=0}^\infty\EE\left[ e^{-\lambda_1 \tau_l}\left(\sum_{k=0}^\infty e^{-\lambda_1\tau_{k}}\right)^{q-1}\right]<\infty.
\end{align}

Using Cauchy-Schwartz inequality, we have that
\begin{align*}
\EE&\left[e^{-\lambda_1 \tau_l}\left(\sum_{k=0}^\infty e^{-\lambda_1\tau_k}\right)^{q-1}\right]\\
&\le \left(\EE\left(e^{-2\lambda_1\tau_l}\right)\right)^{1/2}\left(\EE\left[\left(\sum_{k=0}^\infty e^{-\lambda_1\tau_{k}}\right)^{2(q-1)}\right]\right)^{1/2}.
\end{align*}
Since $\{\tau_l: l\ge 1\}$ are jump times from a Poisson process with rate $\theta$, then 
$\EE\left(e^{-2\lambda_1\tau_l}\right) = (\theta/(\theta+2\lambda_1))^l$.
By the reasoning in \cite[Page 485]{athreya:ghosh:sethuraman:2008}, we also have
\[
\EE\left[\left(\sum_{k=0}^\infty e^{-\lambda_1\tau_{k}}\right)^{2(q-1)}\right]
\le\Gamma(2(q-1)+1)\left(\frac{\theta+\lambda_1}{\lambda_1}\right)^{2(q-1)},
\]
which completes the proof of \eqref{eq:final_claim}.
\end{proof}

The claim in \eqref{eq:claim_moment} follows directly from Proposition~\ref{prop:moment_mBI}, since by the embedding construction, we have $Y^{(i)}_j\le 2\le e^{\lambda_1}$ for all $i,j=1,2$.


%
%

\bibliographystyle{spbasic}      
\bibliography{./bibfile_recip.bib}   

%
%

\end{document}